\pdfoutput=1
\pdftrailerid{}
\documentclass[11pt,a4paper]{article}
\usepackage[T1]{fontenc}
\usepackage[utf8]{inputenc}
\usepackage[sc]{mathpazo}
\usepackage{courier}
\usepackage[scaled=0.88]{helvet}
\usepackage[a4paper,margin=2.6cm,headheight=14pt]{geometry}
\usepackage{microtype}
\usepackage{amsmath,amssymb,amsthm,mathtools}
\usepackage[dvipsnames]{xcolor}
\usepackage{booktabs,tabularx,array}
\usepackage{enumitem}
\usepackage{titlesec}
\usepackage{needspace}
\usepackage{fancyhdr}
\usepackage{caption}
\usepackage{listings}
\usepackage{yfonts}
\usepackage{tikz}
\usetikzlibrary{arrows.meta,positioning,calc,decorations.pathreplacing}
\usepackage{pgfplots}
\pgfplotsset{compat=1.16}
\usepackage[hyphens]{url}
\usepackage[colorlinks=true,linkcolor=zred,citecolor=zred,urlcolor=zred,
  pdftitle={Z-Sigil: A Public-Key Cryptosystem with Chained Selection over a Fibre Bundle of Module-Lattice Keys},
  pdfauthor={Andrea Rondelli},pdfsubject={},pdfkeywords={},pdfcreator={},pdfproducer={}]{hyperref}

\definecolor{zred}{RGB}{140,24,24}
\definecolor{zgold}{RGB}{212,160,23}
\definecolor{zblue}{RGB}{38,56,84}
\definecolor{zgrey}{RGB}{236,238,241}
\definecolor{zpink}{RGB}{242,224,226}

\titleformat{\section}{\Large\bfseries\color{zred}}{\thesection}{0.7em}{}[{\color{zred}\titlerule}]
\titleformat{\subsection}{\large\bfseries\color{zred!80!black}}{\thesubsection}{0.7em}{}
\titleformat{\subsubsection}{\normalsize\bfseries}{\thesubsubsection}{0.7em}{}
\titlespacing*{\section}{0pt}{2.2ex plus 1ex}{1.4ex}

\fancypagestyle{plain}{\fancyhf{}\fancyfoot[C]{\footnotesize\color{zred}\thepage}}

\newtheoremstyle{zthm}{1.2ex}{1.2ex}{\itshape}{}{\bfseries\color{zred}}{.}{0.5em}%
  {\thmname{#1}\thmnumber{ #2}\thmnote{ \mdseries(#3)}}
\newtheoremstyle{zdef}{1.2ex}{1.2ex}{}{}{\bfseries\color{zred}}{.}{0.5em}%
  {\thmname{#1}\thmnumber{ #2}\thmnote{ \mdseries(#3)}}
\theoremstyle{zthm}
\newtheorem{theorem}{Theorem}[section]
\newtheorem{proposition}[theorem]{Proposition}
\newtheorem{lemma}[theorem]{Lemma}
\newtheorem{corollary}[theorem]{Corollary}
\theoremstyle{zdef}
\newtheorem{definition}[theorem]{Definition}
\newtheorem{criterion}[theorem]{Criterion}
\newtheorem{remark}[theorem]{Remark}

\newenvironment{inwords}{\par\medskip\begingroup\small\leftskip=1.5em\rightskip=1.5em\noindent\textsc{In words.}\ }{\par\endgroup\medskip}
\newcommand{\zbox}[1]{\fcolorbox{zred}{white}{\ensuremath{\displaystyle #1}}}
\newcommand{\lab}[1]{\textsf{\small #1}}
\newcommand{\Rq}{R_q}
\newcommand{\ZZ}{\mathbb{Z}}
\newcommand{\RR}{\mathbb{R}}
\newcommand{\CC}{\mathbb{C}}
\newcommand{\E}{\mathbb{E}}
\newcommand{\EP}{\mathcal{E}_P}
\newcommand{\Ahat}{\widehat{A}}
\newcommand{\Tor}{\mathcal{T}}
\newcommand{\ip}[2]{\langle #1,#2\rangle}
\newcommand{\Enc}{\mathfrak{E}}
\newcommand{\Dec}{\mathfrak{D}}
\newcommand{\Kgen}{\mathfrak{K}}
\newcommand{\spec}{\texttt{Z-Sigil/spec-v2}}
\newcommand{\CBD}{\mathrm{CBD}}
\newcommand{\Decode}{\mathrm{Decode}}
\newcommand{\Hol}{\mathrm{Hol}}
\setlist{itemsep=0.3ex,topsep=0.6ex}
\begin{document}
\raggedbottom
\thispagestyle{plain}
\begin{center}
{\fontsize{30}{36}\selectfont\color{zred}Z-Sigil}\par\medskip
{\Large\bfseries A Public-Key Cryptosystem with Chained Selection\\ over a Fibre Bundle of Module-Lattice Keys\par}
\bigskip
{\large Andrea Rondelli}\par\smallskip
Z-Sigil SRLS, Emilia-Romagna, Italy\par\smallskip
{\small\href{https://www.zsigil.com}{\texttt{https://www.zsigil.com}} \ $\cdot$\ \href{mailto:andrearondelli85@gmail.com}{\texttt{andrearondelli85@gmail.com}}}\par\smallskip
29 September 2026
\end{center}
\vspace{1.5em}

\begin{center}\textbf{Abstract}\end{center}
\vspace{-0.5ex}
{\small\leftskip=1.2em\rightskip=1.2em
Z-Sigil is a public-key cryptosystem in which the plaintext chooses which member of a fixed family of module-lattice keys is used next. A message is represented as bytes, prefixed with its length, zero-padded and divided into 32-byte (256-bit) blocks. For each key index $t$, key generation samples a small secret vector $s_t$ and a small error vector $e_t$; the public vector is $b_t=As_t+e_t$, where $A$ is a shared public matrix and arithmetic takes place in a polynomial ring modulo a public integer~$q$.

Geometrically, the indices label a finite set of torsion points of a flat Kähler torus. The secret family assigns one vector to each point, forming a section of the key bundle; $A$ acts within each fibre. A fresh public random value (nonce) initialises a hash state. Each block uses the state to select its key and derive a bit mask. The sender then updates the state using the plaintext block; the receiver performs the same update after recovering that block. Thus the processed plaintext stream, together with the nonce, determines a discrete walk on the torus, and decryption reads the secret section along that walk.

We specify key generation, encryption and decryption, and prove correctness under an explicit noise bound. Under stated decisional Module-LWE assumptions, we prove confidentiality against chosen-plaintext attacks (IND-CPA) for the complete chain, allowing messages chosen after the public key without modelling the state hash as a random oracle. For the example with 16 keys and 64 transmitted blocks, we obtain a decoding-failure upper bound below $2^{-192}$. This is a reliability bound, not a security-level estimate. A restricted model with independent uniform selectors gives exact laws for the initial blocks recovered using an exposed subset of keys. Known-plaintext and candidate-message attacks show why these laws do not bound unrestricted adversaries. The scheme provides no authentication or chosen-ciphertext security.

The construction replaces an earlier geometric proposal that exposed a plaintext scalar along a public direction. The geometric analysis retains both smallness constraints, gives integral-transport obstructions and a sufficient noise budget. Two appendices quantify restricted fragment recovery and develop testable directions for curved, nontrivial bundles. The byte profile, pseudocode and known-answer vectors specify the construction, and implementation records are kept separate from the independent mathematical checks of this revision.

\smallskip\noindent\textit{Specification profile:} \spec{} (Appendix~A).\par}

\vfill
\begin{flushright}
\begin{minipage}{0.62\textwidth}\raggedleft\small
{\color{zred!78!black}\large\textfrak{\ldots\ gemma prezios:is:s:ima\ \ldots\\ \`E cinta da grandi e alte mura con dodici porte\ \ldots}}\\[0.6ex]
\textsc{Apocalisse} 21,11--12
\end{minipage}
\end{flushright}

\newpage
{\hypersetup{linkcolor=black}\small\setlength{\parskip}{0pt}\tableofcontents}
\newpage

\Needspace{7\baselineskip}
\section{Introduction}
A public-key cryptosystem ordinarily fixes one secret and uses it for every block of every message. Z-Sigil fixes a finite family of secrets and lets the plaintext decide the order in which its members are accessed. The family is generated before any message exists; neither the message nor the trajectory creates new keys. State is local to one message and is reinitialised from a public nonce, so sender and receiver need no persistent synchronised state across messages.

Three components have separate roles. Module-LWE supplies the noisy public relation. A fibre bundle over a finite set of torsion points of a flat torus organises the indexed secrets as one section. A plaintext-fed hash chain supplies the walk along which that section is read. Keeping these roles distinct makes it possible to state exactly which properties follow from the arithmetic, which from a chosen recovery procedure, and which remain research questions.

\subsection{One message, end to end}
Frame a byte message by prepending its eight-byte length, then append the minimum zero padding needed to form 32-byte blocks. Optionally prepend $\nu$ independent random 32-byte blocks. Write the complete stream as $m_0,\dots,m_{L-1}$, where $L=\nu+\ell$ and $\ell$ counts the framed-message blocks. The public header carries $L$, the parameters and a fresh nonce $iv$.

Set $x_{-1}=H(\lab{INIT},iv)$. Before processing block $i$, compute
\begin{equation}\label{eq:select}
t_i=H(\lab{FIBRE},x_{i-1}),\qquad \kappa_i=H(\lab{MASK},x_{i-1}),\qquad \mu_i=m_i\oplus\kappa_i.
\end{equation}
Encrypt $\mu_i$ under the public vector $b_{t_i}$ with the block map $\Enc_i$ (Section~\ref{sec:blockmaps}). Once $m_i$ is available, advance the state:
\begin{equation}\label{eq:chain}
x_i=H(\lab{CHAIN},i,x_{i-1},m_i).
\end{equation}
The sender already knows $m_i$; the receiver obtains it by decryption. Both therefore perform the same update, and
\[
(m_{i-1},x_{i-2})\longmapsto x_{i-1}\longmapsto(t_i,\kappa_i)\longmapsto m_i,\qquad (x_{i-1},m_i)\longmapsto x_i .
\]
Block $m_i$ affects $t_{i+1}$, not $t_i$: making $t_i$ depend on the unknown $m_i$ would create a circular decryption rule. The initial state, first selector and first mask are public.

Geometrically, $t_i$ names a point $\theta_{t_i}$ of the base of the key bundle, and the key used at block $i$ is the value of the secret section at that point. The message thus traces a walk $\theta_{t_0}\to\theta_{t_1}\to\cdots$ on the base (Figure~\ref{fig:bundle}); the selectors may revisit a point arbitrarily often, but every block uses fresh encryption randomness. The trajectory changes access order, not the key-generation distribution.

\subsection{Motivation and relation to earlier ideas}
Quantum superposition alone is not a speedup: useful quantum algorithms combine coherent evaluation with interference. Shor's algorithm and amplitude amplification motivate asking whether plaintext-dependent access can impose useful computational dependencies; Manin's exposition provides the conceptual background~\cite{manin-shor}. This is a motivation for studying the architecture, not an argument that it prevents quantum attacks.

Online ciphers process plaintext prefixes and include hash-based chaining constructions~\cite{online}. Key privacy asks whether ciphertexts reveal the recipient key~\cite{keypriv}, while KEM combiners combine encapsulation mechanisms under specified robustness conditions~\cite{combiners}. These are distinct definitions and do not in themselves specify Z-Sigil's coupling of a plaintext-fed state with a pre-generated lattice-key family. We study that coupling here. Establishing novelty relative to all stateful encryption constructions would require a broader prior-art analysis; the claims below are the explicit results of this manuscript.

\begin{figure}[t]
\centering
\begin{tikzpicture}[box/.style={draw=zred,rounded corners=2pt,align=center,font=\small,inner sep=4pt,minimum height=1.25cm,text width=2.75cm},>=Stealth]
\node[box] (a) {state $x_{i-1}$\\ from block $i-1$};
\node[box,right=0.4cm of a] (b) {select fibre $t_i$,\\ derive mask $\kappa_i$};
\node[box,right=0.4cm of b] (c) {decode with $s_{t_i}$,\\ recover $m_i$};
\node[box,right=0.4cm of c,text width=4.1cm] (d) {$x_i=H(\lab{CHAIN},i,x_{i-1},m_i)$\\ names $t_{i+1}$ and $\kappa_{i+1}$};
\draw[->,zred,thick] (a)--(b); \draw[->,zred,thick] (b)--(c); \draw[->,zred,thick] (c)--(d);
\draw[->,zred,thick,dashed] (a.south) -- ++(0,-0.45) -| (d.south);
\end{tikzpicture}
\caption{The one-block delay. The update consumes the old state and the newly recovered plaintext. The arrows describe the prescribed evaluation procedure; they do not establish a lower bound for other algorithms (Section~\ref{sec:serial}).}
\label{fig:delay}
\end{figure}
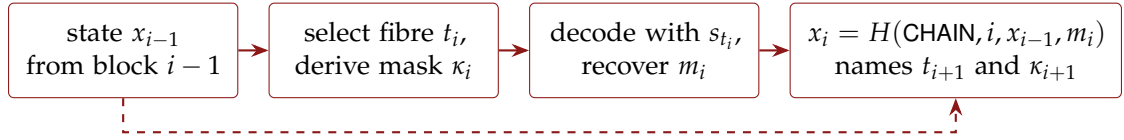

The present scheme is the successor of an earlier geometry-motivated proposal~\cite{rondelli1}. That proposal failed because it exposed a plaintext scalar along a public direction, and Section~\ref{sec:earlier} states the failure and how the present design removes each of its causes.

\subsection{Contributions and claim boundaries}
\begin{enumerate}
\item A geometric setting in which the private key is one section of a discrete fibre bundle of modules over torsion points of a flat Kähler torus, the public key is its noisy image under a fibrewise endomorphism, and a message is a walk on the base (Section~\ref{sec:geom}).
\item Separate, complete specifications of key generation (Section~\ref{sec:keygen}) and of encryption and decryption (Section~\ref{sec:encdec}), with explicit block maps and a chained-inversion proposition.
\item A key-family-conditioned decoding-failure bound (Section~\ref{sec:correct}).
\item A standard-model conditional IND-CPA reduction for the full chain, allowing messages chosen after the public key, with two routes for the per-block step (Section~\ref{sec:security}). This is not a claim of key-dependent-message security.
\item Exact prefix, padding and multiple-stream laws in a restricted recovery model, with explicit counterexamples to stronger exposure interpretations (Section~\ref{sec:exposure}).
\item An augmented-lattice interpretation, a parallel candidate-table algorithm, explicit costs, an integrality obstruction for norm-preserving transport and a quantitative noise budget (Sections~\ref{sec:lattice}, \ref{sec:serial}--\ref{sec:transport}).
\item A byte profile, known-answer vectors, retained prototype test records and independently recomputed numerical figures (Section~\ref{sec:repro}; Appendices~\ref{app:profile}--\ref{app:checklist}).
\end{enumerate}
The scheme has no authentication or chosen-ciphertext security. Its parameters have not been assigned a concrete classical or quantum security level. A small upper bound on decoding failure is not a lower bound on attack cost. The flat geometry supplies a representation, not an additional hardness assumption.

\subsection{Notation and reading guide}
The scheme consists of three algorithms, written in Gothic letters: key generation $\Kgen$, encryption $\Enc$ and decryption $\Dec$,
\[
(pk,sk)\leftarrow\Kgen(T),\qquad C\leftarrow\Enc_{pk}(msg),\qquad \widehat{msg}=\Dec_{sk}(C).
\]
At block level they act through the block maps $\Enc_i$ and $\Dec_i$ of Section~\ref{sec:blockmaps}: $c_i=\Enc_i(m_i)$ and $\hat m_i=\Dec_i(c_i)$. Table~\ref{tab:notation} collects the symbols used throughout. Only $A$ is capitalised among the key components. Norms refer to specified integer coefficient representatives, not abstract residues.

\begin{table}[htbp]
\centering\small
\caption{Notation.}\label{tab:notation}
\begin{tabularx}{\textwidth}{@{}>{\raggedright\arraybackslash}p{4.1cm}X@{}}
\toprule
Symbol & Meaning\\
\midrule
\multicolumn{2}{@{}l}{\emph{Algorithms}}\\
$\Kgen$, $\Enc$, $\Dec$ & key generation, encryption, decryption\\
$\Enc_i$, $\Dec_i$ & block maps at position $i$ (Section~\ref{sec:blockmaps})\\
\midrule
\multicolumn{2}{@{}l}{\emph{Keys}}\\
$pk=(A,b)$ & public key: shared matrix $A\in\Rq^{k\times k}$ and public vectors $b=(b_t)_{t<T}$\\
$sk=s=(s_t)_{t<T}$ & private key: the small secret vectors\\
$e=(e_t)_{t<T}$ & key-generation errors, with $b_t=As_t+e_t$\\
$\sigma$, $\epsilon$, $\beta$ & secret, error and public sections of the key bundle, $\beta=\Ahat\circ\sigma+\epsilon$\\
$\Ahat$ & fibrewise endomorphism defined by $A$\\
\midrule
\multicolumn{2}{@{}l}{\emph{Geometry}}\\
$\Tor$, $\Tor[q]$ & flat Kähler torus and its $q$-torsion points\\
$P=\{\theta_t\}$, $\EP$ & finite base and key bundle $\EP=P\times\Rq^k$\\
$\gamma$ & walk of a message on the base, $\gamma(i)=\theta_{t_i}$\\
\midrule
\multicolumn{2}{@{}l}{\emph{Messages and ciphertexts}}\\
$m_i$, $\kappa_i$, $\mu_i$ & plaintext block, mask and masked block, $\mu_i=m_i\oplus\kappa_i$\\
$x_i$, $t_i$, $iv$ & chain state, selected key index, public nonce\\
$r,f\in\Rq^k$, $g\in\Rq$ & fresh encryption coins and errors of one block\\
$c_i=(u_i,v_i)$ & ciphertext block\\
$C=(\mathrm{hdr},c_0,\dots,c_{L-1})$ & ciphertext; the header $\mathrm{hdr}$ carries parameters, $L$ and $iv$\\
$L$, $\ell$, $\nu$ & transmitted blocks, framed-message blocks, leading random blocks\\
\midrule
\multicolumn{2}{@{}l}{\emph{Parameters and analysis}}\\
$n,k,q,\eta,T$ & ring dimension, module rank, modulus, CBD parameter, number of keys\\
$\Delta=\lfloor q/2\rfloor$, $\delta$ & bit representative, unreduced decoding error\\
$S_t=\|s_t\|_2^2+\|e_t\|_2^2$ & squared length of key $t$\\
$\varepsilon_1,\varepsilon_2,\varepsilon_3$ & Module-LWE advantages (Definition~\ref{def:mlwe}); distinct from the error section $\epsilon$\\
$K$, $p=|K|/T$, $W$ & exposed key set, exposed fraction, number of observed ciphertexts\\
\bottomrule
\end{tabularx}
\end{table}

The shortest technical reading path is key generation, encryption and decryption, the correctness proof, the security experiment and the byte profile. Readers interested in geometry can read Sections~\ref{sec:geom}, \ref{sec:lattice} and~\ref{sec:transport}; readers interested in the seriality proposal should read Sections~\ref{sec:exposure} and~\ref{sec:serial} together. Short paragraphs marked \textsc{In words} accompany the main results; they are glosses and do not replace the hypotheses.

\Needspace{7\baselineskip}
\section{From the earlier construction to the present one}\label{sec:earlier}
The first proposal~\cite{rondelli1} placed the keys in the tangent bundle of a compact Calabi--Yau manifold and encrypted block $i$ as
\begin{equation}\label{eq:old}
c_i=m_i\,N_{i-1}\,e_i,
\end{equation}
where $m_i$ is the plaintext block read as an integer, $e_i\in T_{p_i}\mathcal{M}$ is a public tangent vector and $N_{i-1}$ is a chaining coefficient. A point-dependent groupoid operation $\circledast_p$ entered only decryption. We analysed this construction and found it unusable; the reasons are short and they shaped every choice below.

\begin{proposition}[Collinearity]\label{prop:collinear}
In \eqref{eq:old} the ciphertext is a scalar multiple of the public vector. For every coordinate $j$ with $(e_i)_j\neq0$,
\[
m_iN_{i-1}=(c_i)_j/(e_i)_j,
\]
computable from public data alone, without the point $p_i$, the local basis, the section or $\circledast_{p_i}$.
\end{proposition}
\begin{proof}
Immediate from \eqref{eq:old}. The statement is invariant under change of basis: if $c=\lambda e$ in one basis then $c'=\lambda e'$ in every other, with the same $\lambda$, and the legitimate receiver needs $c_i$ and $e_i$ in a common basis anyway.
\end{proof}

For the first block $N_{-1}=1$, a nonzero public coordinate reveals $m_0$ directly. For integer coordinates, $\gcd_j|(c_0)_j|=|m_0|\gcd_j|(e_0)_j|$; the gcd alone equals $|m_0|$ only for a primitive public vector. Moving the scalar inside a bilinear operation leaves the following exposure mechanism unchanged.

\begin{criterion}[Public scalar exposure]\label{crit:scalar}
If $c=mw$ with a known nonzero vector $w$ over a field, a public coordinate ratio reveals $m$. Bilinearity gives $e\circledast_p(mu)=m(e\circledast_p u)$, so the criterion applies when that direction is publicly computable. Linearity alone, with hidden randomness or noise, is not a general insecurity criterion.
\end{criterion}

Two further observations on the prototype of~\cite{rondelli1} complete the picture. First, in its diagonal realisation the point-dependent coefficients cancel in decryption: $x_aq_a(p)\,(e_aq_a(p))^{-1}=x_a/e_a$. The manifold point was computationally inert. Second, the chaining coefficient depended on the seed, nonce and block index, not on the message; the sequence $(N_i)$ could be precomputed in parallel. Here the plaintext-fed recurrence makes later inputs depend on recovered data in the prescribed loop. This is an algorithmic dependency, not a universal sequential-work lower bound. Known plaintext can supply an update without decrypting its block; other adversarial procedures are analysed in Sections~\ref{sec:exposure}--\ref{sec:serial} and Appendix~\ref{app:seriality}.

Table~\ref{tab:earlier} summarises how the design follows from this analysis. Message bits enter additively and are masked, so Criterion~\ref{crit:scalar} has nothing to act on: the public relation stays noisy and the plaintext never multiplies a public object. The block equations are those of the Lindner--Peikert/Kyber public-key encryption pattern~\cite{lp,kyber-spec,kyber-eurosp}, without ciphertext compression and with all sampling distributions specified. This makes the cancellation and noise calculations transparent; it does not make the result a conforming ML-KEM implementation, which is a standardised KEM with its own construction and analysis~\cite{fips203}. The default rank, ring dimension and modulus resemble ML-KEM-768, but a shared matrix, a family of public samples, uncompressed bulk encryption and plaintext-directed access require their own analysis, and no security estimate is inherited by matching parameters. The object of study is the composition architecture, not a new lattice primitive.

\begin{table}[t]
\centering\small
\caption{What failed in the earlier construction, and what replaces it here.}\label{tab:earlier}
\begin{tabularx}{\textwidth}{@{}XX@{}}
\toprule
Earlier construction~\cite{rondelli1} & Present construction\\
\midrule
plaintext as a scalar factor, $c_i=m_iN_{i-1}e_i$ & plaintext selects one of two cosets, $\Delta\mu$, added to a noisy inner product \eqref{eq:blockmap}\\
security from an ad hoc assumption, refuted by Proposition~\ref{prop:collinear} & two explicit decisional Module-LWE assumptions (Definition~\ref{def:mlwe})\\
manifold point cancels in decryption & base point decides which fibre, hence which secret, is read (Section~\ref{sec:geom})\\
chain independent of the message & chain fed by the recovered plaintext, \eqref{eq:chain}\\
invertibility of ring elements required & no invertibility used in encryption, decryption or correctness (Theorem~\ref{thm:cancel})\\
\bottomrule
\end{tabularx}
\end{table}

\Needspace{7\baselineskip}
\section{Geometric setting: a fibre bundle of keys}\label{sec:geom}
\subsection{A flat Kähler torus and its finite arithmetic}
Let $n=2d$ and let $\Lambda_0=\sum_{j=1}^n\ZZ c_j\subset\RR^n$ be a full-rank lattice. Choose $\RR^n\cong\CC^d$ and set $\Tor=\RR^n/\Lambda_0$. Translations preserve the Euclidean Hermitian metric and the forms
\[
\omega=\frac{i}{2}\sum_{j=1}^d dz_j\wedge d\bar z_j,\qquad \Omega=dz_1\wedge\cdots\wedge dz_d .
\]
They descend to the quotient. Constant coefficients give $d\omega=0$, the metric is flat and hence Ricci-flat, and $\Omega$ trivialises the canonical bundle. We use the term \emph{flat Kähler torus}, avoiding competing Calabi--Yau conventions (the torus has trivial holonomy, not $SU(d)$).

For an odd prime $q$ and a power of two $n\ge2$, put
\begin{equation}\label{eq:ring}
\Rq=\ZZ_q[X]/(X^n+1).
\end{equation}
A basis-dependent additive isomorphism onto the $q$-torsion of the torus is
\begin{equation}\label{eq:iota}
\iota:\ \sum_{j=0}^{n-1}a_jX^j\longmapsto\frac1q\sum_{j=0}^{n-1}a_jc_{j+1}+\Lambda_0\in\Tor[q]=q^{-1}\Lambda_0/\Lambda_0 .
\end{equation}
Changing a coefficient representative by a multiple of $q$ does not change the coset, so $\iota$ is well defined; a point of $\Tor[q]$ is literally a coset of $\Lambda_0$. Ring multiplication is additional structure supplied by \eqref{eq:ring}; it is not induced by the flat metric. Manin's work on theta functions and quantum tori, where finite level structures and Heisenberg groups carry the arithmetic of a torus, provides related context~\cite{manin-theta}, not a security reduction for this construction.

\subsection{The key bundle, its sections and the walk}
\begin{definition}[Key bundle]\label{def:bundle}
Choose distinct points $P=\{\theta_t:t\in[T]\}\subset\Tor[q]$ with $1\le T\le q^n$ and a module rank $k\ge1$. The \emph{key bundle} is
\[
\pi:\ \EP=P\times\Rq^k\longrightarrow P,\qquad(\theta,v)\mapsto\theta,
\]
a fibre bundle of free $\Rq$-modules of rank $k$ over the finite base $P$, with fibre $\EP|_{\theta_t}=\{\theta_t\}\times\Rq^k$.

Three sections of $\EP$ carry the key material:
\[
\sigma(\theta_t)=(\theta_t,s_t),\qquad \epsilon(\theta_t)=(\theta_t,e_t),\qquad \beta(\theta_t)=(\theta_t,b_t).
\]
The \emph{secret section} $\sigma$ is the private key; $\epsilon$ is the key-generation error, and $\beta$ is the \emph{public section}. The shared matrix defines a fibrewise endomorphism
\[
\Ahat:\ \EP\to\EP,\qquad(\theta,v)\mapsto(\theta,Av),
\]
which covers the identity of $P$: it acts inside every fibre and never moves the base point. In this language the whole public key is one identity between sections,
\begin{equation}\label{eq:sections}
\zbox{\beta=\Ahat\circ\sigma+\epsilon,}
\end{equation}
read fibre by fibre as $b_t=As_t+e_t$. The public section is a noisy image of the secret section; recovering $\sigma$ from $\beta$ is small-secret Module-LWE in every fibre (Section~\ref{sec:lattice}).
\end{definition}

\begin{definition}[Walk and pulled-back key]\label{def:walk}
A nonce and a stream $m_0,\dots,m_{L-1}$ determine, through \eqref{eq:select}--\eqref{eq:chain}, a map
\[
\gamma:\{0,\dots,L-1\}\to P,\qquad\gamma(i)=\theta_{t_i},
\]
the \emph{walk} of the message. The pulled-back bundle has fibre $\EP|_{\gamma(i)}$ at $i$, and its induced section is $i\mapsto(i,\sigma(\gamma(i)))$. In the product coordinates, its secret-vector component is $s_{t_i}$, the key used at block $i$.
\end{definition}

The division of labour is now visible. The section $\sigma$ is sampled once and never changes; the plaintext decides the walk $\gamma$ and nothing else. Selection moves the accessed base point from $\theta_{t_i}$ to $\theta_{t_{i+1}}$ without transporting a fibre vector, while $\Ahat$ acts inside fibres. Since $\gamma(i+1)$ depends on $m_i$ through \eqref{eq:chain}, the prescribed receiver computes the next point after recovering the current block. Known or guessed plaintext can also supply this update; the recurrence does not establish absolute hiding of the walk or a lower bound for every algorithm.

\begin{remark}[What the geometry is, and is not, here]
Over a finite discrete base every bundle is a product, and $\EP$ is written as one; no continuous extension of $\sigma$, no tangent-bundle interpretation and no connection are assumed. Relabelling $P$ leaves the cryptographic algorithm unchanged, and the formal variable $X$ is not a coordinate or point of $\Tor$. What is not trivial is the walk: it is chosen by data the receiver learns one block at a time. Section~\ref{sec:transport} states what additional structure, transport along the edges of the base, would turn this organisation into genuinely path-dependent geometry, and what it would have to satisfy. The augmented lattice of Section~\ref{sec:lattice} has a different dimension and purpose from $\Lambda_0$.
\end{remark}

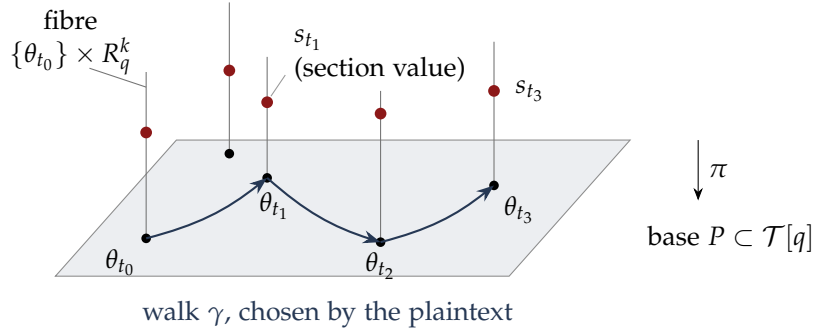
\begin{figure}[t]
\centering
\begin{tikzpicture}[x=1cm,y=1cm,>=Stealth]
\fill[zgrey] (0,0) -- (6,0) -- (7.6,1.8) -- (1.6,1.8) -- cycle;
\draw[black!45] (0,0) -- (6,0) -- (7.6,1.8) -- (1.6,1.8) -- cycle;
\coordinate (p0) at (1.2,0.5); \coordinate (p1) at (2.8,1.3); \coordinate (p2) at (4.3,0.45);
\coordinate (p3) at (5.8,1.2); \coordinate (p4) at (2.3,1.62);
\foreach \p/\h in {p0/2.2,p1/1.6,p2/2.0,p3/1.9,p4/2.0}{\draw[black!55] (\p) -- ++(0,\h);}
\foreach \p/\y in {p0/1.4,p1/1.0,p2/1.7,p3/1.25,p4/1.1}{\fill[zred] ($(\p)+(0,\y)$) circle (2.2pt);}
\foreach \p in {p0,p1,p2,p3,p4}{\fill[black] (\p) circle (1.8pt);}
\draw[->,zblue,thick] (p0) to[bend right=12] (p1);
\draw[->,zblue,thick] (p1) to[bend right=12] (p2);
\draw[->,zblue,thick] (p2) to[bend right=12] (p3);
\node[below left,font=\small] at (p0) {$\theta_{t_0}$};
\node[below,font=\small] at ($(p1)+(0.1,-0.05)$) {$\theta_{t_1}$};
\node[below,font=\small] at (p2) {$\theta_{t_2}$};
\node[below right,font=\small] at (p3) {$\theta_{t_3}$};
\node[font=\small,anchor=west,align=left] at ($(p1)+(0.22,1.62)$) {$s_{t_1}$\\ (section value)};\draw[black!40] ($(p1)+(0.05,1.08)$) -- ($(p1)+(0.24,1.36)$);
\node[font=\small,anchor=west] at ($(p3)+(0.15,1.25)$) {$s_{t_3}$};
\node[font=\small,align=center] at (0.2,3.1) {fibre\\ $\{\theta_{t_0}\}\times R_q^k$};
\draw[black!55] (0.45,2.75) -- ($(p0)+(0,2.0)$);
\node[font=\small,anchor=west] at (7.7,0.5) {base $P\subset\Tor[q]$};
\node[font=\small,zblue] at (3.6,-0.5) {walk $\gamma$, chosen by the plaintext};
\draw[->] (8.5,1.8) -- (8.5,1.0) node[midway,right,font=\small] {$\pi$};
\end{tikzpicture}
\caption{The key bundle over a finite set $P$ of $q$-torsion points of the flat torus (shaded fundamental domain, opposite sides identified). Above each point of $P$ stands a fibre $\{\theta_t\}\times\Rq^k$; the red dots are the values of the secret section $\sigma$. The plaintext-fed chain draws a walk $\gamma$ on the base, and block $i$ is decrypted with $\sigma(\gamma(i))$. The section is defined at the points of $P$ and nowhere else; the fibre at the upper left is not visited by this message.}
\label{fig:bundle}
\end{figure}

\subsection{Negacyclic arithmetic}
For $a,b\in\Rq$, coefficients multiply according to
\begin{equation}\label{eq:negacyclic}
(ab)_c=\sum_{i+j=c}a_ib_j-\sum_{i+j=c+n}a_ib_j\pmod q,\qquad0\le c<n,
\end{equation}
where $0\le i,j<n$. The minus sign is $X^n=-1$. For module vectors, $\ip ab=\sum_{j=1}^k a_jb_j$ is $\Rq$-bilinear, with no complex conjugation. The ring is commutative and generally not a field. In particular $\ip s{A^\top r}=\ip{As}r$ without requiring $A$ to be invertible.

\begin{table}[t]
\centering\small
\caption{The objects used in one block. All fresh samples are independent of the long-term family.}
\begin{tabular}{@{}lll@{}}
\toprule
Object & Space & Role\\
\midrule
$A$ & $\Rq^{k\times k}$ & shared public matrix; fibrewise endomorphism $\Ahat$\\
$b_t$ & $\Rq^k$ & public section at $\theta_t$\\
$s_t,e_t$ & $\Rq^k$ & secret section and key-generation error at $\theta_t$\\
$r,f$ & $\Rq^k$ & fresh coins and encryption error\\
$g$ & $\Rq$ & fresh encryption error\\
$(u,v)$ & $\Rq^k\times\Rq$ & ciphertext pair\\
$m_i,\mu_i,\kappa_i$ & $\{0,1\}^n$ & plaintext, masked block and mask\\
\bottomrule
\end{tabular}
\end{table}

\Needspace{7\baselineskip}
\section{Key generation}\label{sec:keygen}
\subsection{Parameters and noise}
The numerical profile fixes
\[
n=256,\qquad k=3,\qquad q=3329,\qquad\eta=2,
\]
with $T=16$ and $\nu=0$ by default. A $\CBD_\eta$ coefficient is the difference of two independent sums of $\eta$ unbiased bits. Its mean is zero, its variance $\eta/2$ and its support $[-\eta,\eta]$. All coefficients of each fresh small polynomial or vector are sampled independently, then reduced modulo $q$. The symbolic formulas apply to other valid choices, subject to a fresh correctness and security analysis.

$H$ denotes a typed, domain-separated public interface: \lab{INIT} and \lab{CHAIN} output 256-bit states, \lab{MASK} outputs $n$ bits and \lab{FIBRE} outputs an index in $[T]$. These are separate domains, not a reused digest; Appendix~\ref{app:profile} fixes the concrete instantiation. The hash is deterministic, and exact uniformity of selectors is an assumption used only in the ideal exposure calculations of Section~\ref{sec:exposure}, not a proved property of every real trajectory.

\subsection{The key-generation algorithm}
\begin{definition}[Key generation $\Kgen$]\label{def:keygen}
The algorithm $(pk,sk)\leftarrow\Kgen(T)$ is as follows. Sample $A\leftarrow\Rq^{k\times k}$ uniformly. Independently for each $t\in[T]$, sample small $s_t,e_t\in\Rq^k$ and set
\begin{equation}\label{eq:keygen}
b_t=As_t+e_t,\qquad pk=(A,b),\qquad sk=s.
\end{equation}
\end{definition}

In the language of Section~3.2, key generation samples a random small section $\sigma$ of $\EP$, an independent small error section $\epsilon$, a uniform fibrewise endomorphism $\Ahat$, and publishes $\Ahat$ together with the public section $\beta=\Ahat\circ\sigma+\epsilon$ of \eqref{eq:sections}. The values of $\sigma$ at different base points are independent, so the family consists of $T$ Module-LWE samples that are independent conditional on the shared matrix $A$; Section~\ref{sec:storage} and item~4 of Section~\ref{sec:roadmap} discuss the alternative of an independent $A_t$ per fibre.

The errors may be discarded after setup: knowledge of $s_t$ permits their reconstruction modulo $q$ as $b_t-As_t$, and they are never fresh per-message coins. The error is nonetheless essential to the hiding relation and cannot be dropped: without it, $b_t=As_t$ is a noiseless linear image and, for invertible $A$, reveals $s_t=A^{-1}b_t$, another elementary linear-algebraic exposure.

\par\medskip\noindent\begin{minipage}{\linewidth}
\begin{lstlisting}[caption={Key generation $\mathfrak{K}$; Small means fresh independent $\mathrm{CBD}_\eta$ samples.},escapeinside={``}]
KeyGen(T):
  A <- Uniform(Rq^(k x k))           # rejection-sampled, Appendix A
  for t = 0,...,T-1:                  # independently across t
      s[t], e[t] <- Small(Rq^k), Small(Rq^k)
      b[t] <- A*s[t] + e[t]
  return pk = (A, b), sk = s          # the e[t] are discarded
\end{lstlisting}
\end{minipage}\par\medskip

At the default parameters a polynomial occupies 384 bytes, so the public key $(A,b_0,\dots,b_{15})$ is 21\,888 bytes and the private key $(s_0,\dots,s_{15})$ is 18\,432 bytes in the residue format of Appendix~\ref{app:profile}; Section~\ref{sec:storage} gives the general formulas.

\Needspace{7\baselineskip}
\section{Encryption and decryption}\label{sec:encdec}
\subsection{Encoding and decoding}
A bit polynomial $\mu=\sum_{j<n}\mu_jX^j$ with $\mu_j\in\{0,1\}$ is encoded as $\Delta\mu$, where $\Delta=\lfloor q/2\rfloor$. A residue is decoded by taking its centred representative in $(-q/2,q/2]$ and outputting one exactly when its absolute value exceeds $q/4$. Geometrically each coefficient lives on the circle $\ZZ_q$ of residues; the two bit representatives $0$ and $\Delta$ sit almost half a circle apart, and decoding chooses the nearer one (Figure~\ref{fig:decode}). Message bits enter additively through $\Delta\mu$: they select two coefficient representatives rather than scale a public key vector.

\subsection{The block maps}\label{sec:blockmaps}
For each block, derive $t_i,\kappa_i,\mu_i$ by \eqref{eq:select}. Suppressing $i$ and putting $t=t_i$, sample fresh independent small $r,f\in\Rq^k$ and $g\in\Rq$, independently of one another and of key generation, and output
\begin{equation}\label{eq:blockmap}
\zbox{u=A^\top r+f,\qquad v=\ip{b_t}r+g+\Delta\mu.}
\end{equation}
The randomness $r,f,g$ is neither transmitted nor recovered, and it must be fresh even when the same index is revisited. Here $f,g$ are temporary encryption errors, while $e_t$ always means the key-generation error. The ciphertext is the header followed by the pairs $(u_i,v_i)$. At the receiver, compute
\begin{equation}\label{eq:decrypt}
\zbox{w=v-\ip{s_t}u,\qquad\hat\mu=\Decode(w),\qquad\hat m_i=\hat\mu\oplus\kappa_i,}
\end{equation}
then use $\hat m_i$ in \eqref{eq:chain}. Here $w$ is a noisy encoding, not the final plaintext. The two operations are the state-dependent, randomised block maps
\[
\zbox{\begin{aligned}
\Enc_i(m_i;r,f,g)&=\bigl(A^\top r+f,\ \ip{b_{t_i}}r+g+\Delta(m_i\oplus\kappa_i)\bigr),\\[4pt]
\Dec_i(u,v)&=\Decode\bigl(v-\ip{s_{t_i}}u\bigr)\oplus\kappa_i.
\end{aligned}}
\]
where $t_i$ and $\kappa_i$ are computed from the running state $x_{i-1}$. Thus the ciphertext block is $c_i=(u_i,v_i)=\Enc_i(m_i;r,f,g)$ and the receiver recovers $\hat m_i=\Dec_i(c_i)$. The Gothic letters $\Enc_i,\Dec_i$ denote block maps at position $i$, never the bundle $\EP$. Composed along the chain, they define the encryption and decryption algorithms of the whole message,
\[
\Enc_{pk}(msg)=C=(\mathrm{hdr},c_0,\dots,c_{L-1}),\qquad\Dec_{sk}(C)=\widehat{msg},
\]
specified in Listings~\ref{lst:enc} and~\ref{lst:dec}. Correct recovery, $\widehat{msg}=msg$, holds under the noise condition of Proposition~\ref{prop:chaininv}; malformed inputs may be rejected.

\begin{proposition}[Block inversion]\label{prop:blockinv}
Fix a state $x_{i-1}$ and a block $m_i$. If every coefficient of the unreduced error $\delta=\ip{e_{t_i}}r+g-\ip{s_{t_i}}f$ satisfies $|\delta_c|<q/4-1/2$, then $\Dec_i\bigl(\Enc_i(m_i;r,f,g)\bigr)=m_i$.
\end{proposition}
\begin{proof}
By Theorem~\ref{thm:cancel}, $v-\ip{s_{t_i}}u=\Delta\mu_i+\delta$ and the margin makes $\Decode$ return $\mu_i=m_i\oplus\kappa_i$. Since both maps use the same $\kappa_i$, the XOR returns $m_i$.
\end{proof}

\begin{proposition}[Chained inversion]\label{prop:chaininv}
Suppose every block of a stream satisfies the margin of Proposition~\ref{prop:blockinv}. Then at every position the receiver's state, selector and mask equal the sender's, decryption returns $m_0,\dots,m_{L-1}$, and $\Dec_{sk}\bigl(\Enc_{pk}(msg)\bigr)=msg$.
\end{proposition}
\begin{proof}
By induction on $i$. Both parties set $x_{-1}=H(\lab{INIT},iv)$ from the public nonce. If the receiver holds the sender's $x_{i-1}$, it computes the same $t_i$ and $\kappa_i$, hence decrypts with the key the sender used; Proposition~\ref{prop:blockinv} gives $\hat m_i=m_i$, so $x_i=H(\lab{CHAIN},i,x_{i-1},\hat m_i)$ is again the sender's state.
\end{proof}

An error in block $i$ can desynchronise subsequent states and propagate through the suffix. It need not change every later state or plaintext block: collisions or reconvergence are not excluded. Section~\ref{sec:correct} conservatively counts the first decoding error as a failure of exact whole-stream recovery; before that error the receiver follows the sender's walk.

\subsection{Framing and leading random blocks}
For message bytes $msg$, frame $\mathrm{U64}(|msg|)\,\|\,msg$ and append the least zero padding making its length a multiple of~32. The number of framed blocks is
\[
\ell=\left\lceil\frac{8+|msg|}{32}\right\rceil .
\]
Prepending $\nu$ independent random blocks yields $L=\nu+\ell$. All stream bounds include these leading blocks. A 2048-byte payload therefore uses 65 blocks, or 73 if $\nu=8$. The recipient discards the first $\nu$ recovered blocks and checks the exact length and canonical padding of the rest.

The header carries $L$, but \spec{} does not include $L$ in the hash context. The hash binds $(n,k,q,\eta,T,\nu)$, and \lab{CHAIN} also binds the position $i$. Header and frame checks are syntax checks, not authentication, and no statement below relies on cryptographic binding of $L$.

\begin{figure}[t]
\centering
\begin{tikzpicture}[x=0.78cm,y=0.78cm]
\def\orig{{14,2,3,13,3,5,1,14}}
\def\flip{{14,2,3,13,7,7,13,1}}
\foreach \i in {0,...,7}{
  \pgfmathsetmacro{\a}{\orig[\i]} \pgfmathsetmacro{\b}{\flip[\i]}
  \ifnum\i<4 \def\cc{zgrey}\else\def\cc{zpink}\fi
  \fill[\cc] (1.2*\i,1.1) rectangle ++(1.05,0.9); \node at (1.2*\i+0.525,1.55) {\pgfmathprintnumber{\a}};
  \fill[\cc] (1.2*\i,-0.2) rectangle ++(1.05,0.9); \node at (1.2*\i+0.525,0.25) {\pgfmathprintnumber{\b}};
  \node[font=\scriptsize,black!55] at (1.2*\i+0.525,-0.55) {\i};}
\draw[dashed,black!60] (4.725,-0.8) -- (4.725,2.6);
\node[anchor=east,font=\small] at (-0.15,1.55) {original stream};
\node[anchor=east,font=\small] at (-0.15,0.25) {bit 0 of block 3 flipped};
\node[font=\small,black!70,anchor=east] at (4.55,2.45) {selectors 0--3 unchanged};
\node[font=\small,zred,anchor=west] at (4.9,2.45) {first change at block 4};
\end{tikzpicture}
\caption{A reproducible selector trace at $T=16$, $\nu=0$, with all-zero input blocks except for the indicated flipped bit. These are hash-level fixtures (Appendix~\ref{app:kat}), not framed user messages or a security experiment.}
\label{fig:trace}
\end{figure}

\subsection{Pseudocode and a worked walk}
The listings use the fixed parameters of Appendix~\ref{app:profile}. \texttt{Frame} and \texttt{RandomBlocks} return lists of 32-byte blocks; $\|$ between such lists denotes concatenation. Blocks and masks are identified with coefficient bits in the stated bit order. Ring operations are reduced modulo $q$. \texttt{ParseAndCheck} and \texttt{Unframe} reject invalid syntax as specified in Appendix~\ref{app:profile}; rejection is not authentication. The pairs produced by \texttt{emit} are collected in order before serialisation with the header.
\par\medskip\noindent\begin{minipage}{\linewidth}
\begin{lstlisting}[caption={Encryption $\mathfrak{E}_{pk}$.},label={lst:enc}]
Encrypt((A,b), msg, T, nu):
  blocks <- RandomBlocks(nu) || Frame(msg)
  iv <- Uniform(32 bytes);  x <- H(INIT,iv)
  for i,m in enumerate(blocks):
      t <- H(FIBRE,x);  mask <- H(MASK,x)
      r,f,g <- Small(Rq^k), Small(Rq^k), Small(Rq)
      u <- transpose(A)*r + f
      v <- inner(b[t],r) + g + Delta*(m (+) mask)
      emit (u,v)
      x <- H(CHAIN,i,x,m)                 # old x on the right
  return Header(T,nu,len(blocks),iv) || emitted pairs
\end{lstlisting}
\end{minipage}\par\medskip

\par\medskip\noindent\begin{minipage}{\linewidth}
\begin{lstlisting}[label={lst:dec},caption={Decryption $\mathfrak{D}_{sk}$; parsing is bounded and canonical, and is syntax checking only.}]
Decrypt(s, ciphertext, T, nu):
  iv,pairs <- ParseAndCheck(ciphertext,T,nu)
  x <- H(INIT,iv);  blocks <- []
  for i,(u,v) in enumerate(pairs):
      t <- H(FIBRE,x);  mask <- H(MASK,x)
      m <- Decode(v - inner(s[t],u)) (+) mask
      append m to blocks
      x <- H(CHAIN,i,x,m)
  return Unframe(blocks[nu:])
\end{lstlisting}
\end{minipage}\par\medskip

Figure~\ref{fig:trace} uses the exact hash profile and an all-zero nonce. Flipping bit~0 of block~3 leaves selectors~0 through~3 unchanged and changes the later selectors in this example: the walk on the base is the same up to $\theta_{t_3}$ and then leaves it. The general claim is delayed influence, not that every future index must differ; collisions and repeated indices are possible. Figure~\ref{fig:walk} shows the same mechanism over a longer stream.

\begin{figure}[t]
\centering
\begin{tikzpicture}
\begin{axis}[width=0.95\textwidth,height=6.2cm,xmin=-1,xmax=48,ymin=-1,ymax=17.5,
  xlabel={block index $i$},ylabel={selected fibre $t_i$},ytick={0,4,8,12,15},
  axis lines=left,tick label style={font=\small},label style={font=\small},
  legend style={font=\scriptsize,draw=none,at={(0.98,1.02)},anchor=south east,legend columns=2}]
\fill[zgrey] (axis cs:-1,-1) rectangle (axis cs:11.5,17.5);
\draw[dashed,black!50] (axis cs:11.5,-1) -- (axis cs:11.5,17.5);
\addplot[const plot,zgold,thick,dashed,mark=*,mark size=1.2pt] coordinates {(0,14) (1,2) (2,3) (3,13) (4,3) (5,5) (6,1) (7,14) (8,13) (9,3) (10,0) (11,11) (12,3) (13,7) (14,8) (15,11) (16,11) (17,5) (18,5) (19,10) (20,13) (21,8) (22,6) (23,5) (24,13) (25,14) (26,7) (27,5) (28,7) (29,3) (30,1) (31,5) (32,6) (33,3) (34,15) (35,1) (36,1) (37,6) (38,5) (39,14) (40,14) (41,12) (42,6) (43,4) (44,14) (45,7) (46,7) (47,13)};
\addlegendentry{bit 0 of block 11 flipped}
\addplot[const plot,zred,thick,mark=*,mark size=1.2pt] coordinates {(0,14) (1,2) (2,3) (3,13) (4,3) (5,5) (6,1) (7,14) (8,13) (9,3) (10,0) (11,11) (12,13) (13,11) (14,8) (15,5) (16,0) (17,0) (18,12) (19,3) (20,12) (21,2) (22,2) (23,6) (24,13) (25,15) (26,6) (27,7) (28,6) (29,13) (30,4) (31,14) (32,4) (33,5) (34,13) (35,3) (36,0) (37,3) (38,4) (39,5) (40,15) (41,6) (42,0) (43,15) (44,2) (45,7) (46,13) (47,1)};
\addlegendentry{original stream}
\node[font=\scriptsize,align=center,fill=zgrey] at (axis cs:5.3,16.3) {identical trajectory\\(blocks 0--11)};
\node[font=\scriptsize,zblue] (lab) at (axis cs:26,16.3) {first divergence at block 12};
\draw[->,zblue] (lab.south west) -- (axis cs:12.2,3.4);
\end{axis}
\end{tikzpicture}
\caption{The walk itself, computed with the exact \spec{} hash at $T=16$, $\nu=0$, an all-zero nonce and all-zero input blocks, for 48 blocks. The two streams are identical except for one flipped bit in block~11. Selectors agree up to and including block~11 and separate from block~12 onwards, which is the one-block delay of \eqref{eq:chain} seen over a whole message. Later agreements are ordinary collisions in $[T]$, not persistence of the original walk.}
\label{fig:walk}
\end{figure}
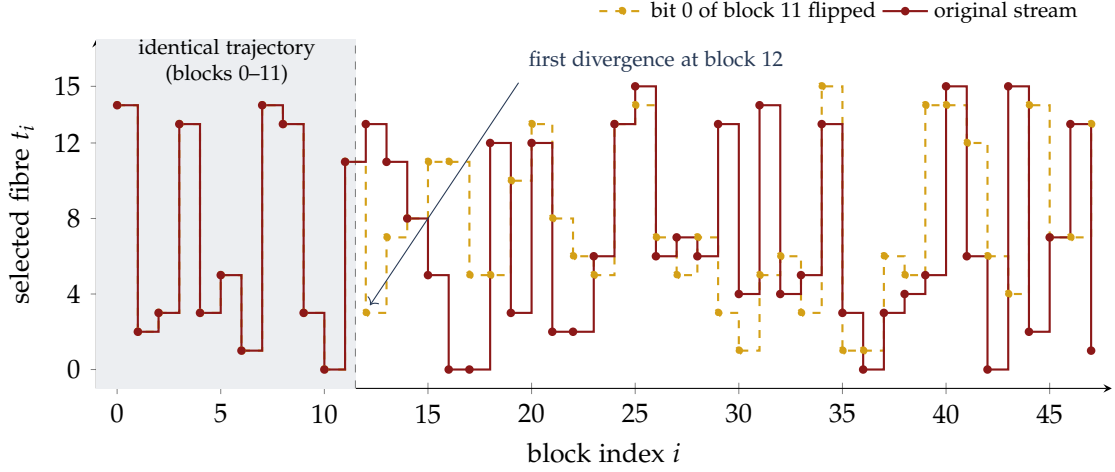

\Needspace{7\baselineskip}
\section{Correctness and decoding failure}\label{sec:correct}
\begin{theorem}[Cancellation and sufficient margin]\label{thm:cancel}
Using unreduced integer negacyclic representatives for the error,
\begin{equation}\label{eq:cancel}
\zbox{w=\Delta\mu+\delta\pmod q,\qquad\delta=\ip{e_t}r+g-\ip{s_t}f.}
\end{equation}
A sufficient coefficientwise condition for correct decoding is
\begin{equation}\label{eq:margin}
|\delta_c|<q/4-1/2 .
\end{equation}
At $q=3329$, $|\delta_c|\le831$ suffices for either bit.
\end{theorem}
\begin{proof}
Bilinearity gives $\ip{s_t}u=\ip{As_t}r+\ip{s_t}f=\ip{b_t-e_t}r+\ip{s_t}f$, which with \eqref{eq:blockmap} yields \eqref{eq:cancel}. For bit zero, \eqref{eq:margin} gives circular distance to zero below $q/4$. For bit one, $\Delta=(q-1)/2$ and \eqref{eq:margin} place $\Delta+\delta_c$ strictly inside $(q/4,\,3q/4-1)$, so its circular distance from zero exceeds $q/4$ and the decoder returns one. Invertibility of $A$ is not used.
\end{proof}

The half-unit adjustment matters because $q$ is odd and the encoding uses a floor: a one bit with $\delta_c=-832$ becomes $1664-832=832$, which is decoded as zero although $832<q/4=832.25$. Figure~\ref{fig:decode} displays this boundary.

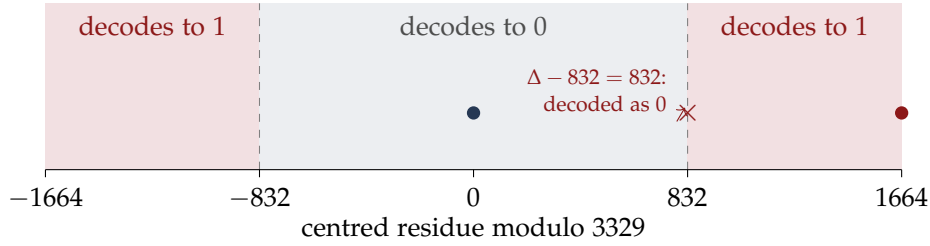
\begin{figure}[t]
\centering
\begin{tikzpicture}[x=0.0034cm,y=1cm]
\fill[zpink] (-1664,0) rectangle (-832,2.2);
\fill[zgrey] (-832,0) rectangle (832,2.2);
\fill[zpink] (832,0) rectangle (1664,2.2);
\draw[dashed,black!50] (-832,0) -- (-832,2.2); \draw[dashed,black!50] (832,0) -- (832,2.2);
\node[zred,font=\small] at (-1248,1.9) {decodes to 1};
\node[black!70,font=\small] at (0,1.9) {decodes to 0};
\node[zred,font=\small] at (1248,1.9) {decodes to 1};
\fill[zblue] (0,0.75) circle (2.5pt);
\fill[zred] (1664,0.75) circle (2.5pt);
\node[zred,font=\large] at (832,0.75) {$\times$};
\node[font=\scriptsize,zred,align=right,anchor=east] (bl) at (790,1.05) {$\Delta-832=832$:\\ decoded as 0};
\draw[->,zred] (bl.south east) -- (822,0.82);
\draw (-1664,0) -- (1664,0);
\foreach \x in {-1664,-832,0,832,1664}{\draw (\x,0) -- (\x,-0.1) node[below,font=\small] {$\x$};}
\node[font=\small] at (0,-0.75) {centred residue modulo 3329};
\end{tikzpicture}
\caption{Exact decoding regions on the circle of residues and the odd-modulus boundary. The sufficient margin concerns the unreduced error, while the decoder sees residues modulo $q$. Here $\Delta=1664$ and $|\delta_c|\le831$; the boundary example is $1664-832=832\mapsto0$. The boundary residues $\pm832$ belong to the region decoded as~0.}
\label{fig:decode}
\end{figure}

\subsection{An averaged bound}
For a fixed fibre and output coefficient, each negacyclic product contains $n$ signed products of disjoint independent coefficient pairs. The two inner products in \eqref{eq:cancel} therefore contribute $2kn$ independent centred terms, each with variance $(\eta/2)^2$ and magnitude at most $\eta^2$; the extra $g_c$ has variance $\eta/2$. Thus
\begin{equation}\label{eq:var}
\E\,\delta_c=0,\qquad\mathrm{Var}(\delta_c)=\frac{kn\eta^2}{2}+\frac{\eta}{2}=1537 .
\end{equation}
Bernstein's inequality with uniform bound~4 gives
\begin{equation}\label{eq:bernstein}
\Pr(|\delta_c|\ge832)\le2\exp\Bigl(-\frac{832^2}{2(1537+4\cdot832/3)}\Bigr)=3.16208\ldots\times10^{-57}.
\end{equation}
For a 64-block trajectory independent of key generation, a union bound over $Ln=16384$ coefficient events gives $5.18075\ldots\times10^{-53}<2^{-173}$. If selection depends on the public key, a union over all $TLn$ hypothetical position/fibre coefficients remains a valid conservative bound, below $2^{-169}$. Conditioning on the whole family gives a substantially sharper result and removes that extra factor~$T$.

\subsection{Conditioning on the realised family}
\begin{lemma}[Fixed-key tail]\label{lem:fixedkey}
Fix integer coefficient representatives $\bar s,\bar e$ with $b=A\bar s+\bar e\pmod q$ and set $S=\|\bar s\|_2^2+\|\bar e\|_2^2$. With fresh independent coefficientwise $\CBD_\eta$ samples $r,f,g$,
\begin{equation}\label{eq:fixedkey}
\mathrm{Var}(\delta_c\mid\bar s,\bar e)=\frac\eta2(S+1),\qquad\Pr(|\delta_c|\ge a\mid\bar s,\bar e)\le2\exp\Bigl(-\frac{a^2}{\eta(S+1)}\Bigr).
\end{equation}
\end{lemma}
\begin{proof}
For $Z\sim\CBD_\eta$, $\E e^{tZ}=\cosh(t/2)^{2\eta}\le e^{\eta t^2/4}$. For a fixed output coefficient $c$, the weights multiplying the entries of $r$ are a signed permutation of all coefficients of $\bar e$, and similarly those of $f$ use all of $\bar s$; the coefficient of $g_c$ is one. These independent samples have squared weights summing to $S+1$, which proves the variance and the bound $\E[e^{t\delta_c}\mid\bar s,\bar e]\le e^{\eta(S+1)t^2/4}$. Chernoff's bound with $t=2a/(\eta(S+1))$, applied to both tails, proves \eqref{eq:fixedkey}.
\end{proof}

\begin{inwords}
Once the key is fixed, the decryption noise is a weighted sum of fresh small samples whose weights are the coefficients of the key. The key enters this upper bound only through its squared length $S$; the $+1$ is the one fresh polynomial $g$ that no key multiplies. This does not say that the whole noise distribution depends only on~$S$.
\end{inwords}

For $\eta=2$ a squared coefficient takes the values $0$, $1$ and $4$ with probabilities $6/16$, $8/16$ and $2/16$. Hence, for $h=2kn=1536$,
\begin{equation}\label{eq:Slaw}
\Pr(S=s)=16^{-h}[z^s](6+8z+2z^4)^h,\qquad\E S=1536,\qquad\mathrm{Var}(S)=2304 .
\end{equation}
This is an exact finite law; its integer coefficients and tails are evaluated without Monte Carlo through the recurrence \eqref{eq:recurrence} of Section~\ref{sec:repro}.

\begin{theorem}[Family-conditioned stream bound]\label{thm:stream}
Let $S_t=\|s_t\|_2^2+\|e_t\|_2^2$ for the sampled integer representatives. The message may be chosen after seeing the public key, but before the fresh independent encryption coins; the nonce and leading random plaintext blocks are independent of those coins. For any integer $S_0>0$ and at most $L$ transmitted blocks,
\begin{equation}\label{eq:streambound}
\Pr(\text{any decoding failure})\le T\Pr(S\ge S_0)+2Ln\exp\Bigl(-\frac{832^2}{\eta S_0}\Bigr).
\end{equation}
At $T=16$, $L=64$ and $S_0=2400$ this is $2^{-192.5614\ldots}<2^{-192}$.
\end{theorem}
\begin{proof}
Let $G$ be the event that every $S_t<S_0$; a union bound over the $T$ independently sampled pairs gives $\Pr(G^c)\le T\Pr(S\ge S_0)$. Condition on the complete key family, the selected message, the nonce and the leading plaintext blocks. This fixes the true encryption walk even when it depends on the public key. The encryption coins remain independent, and on $G$ every selected key has $S_{t_i}+1\le S_0$ by integrality. Lemma~\ref{lem:fixedkey} with $a=832$ therefore bounds each of the $Ln$ actual coefficient events by the second term of \eqref{eq:streambound} divided by $Ln$; by Theorem~\ref{thm:cancel}, avoiding all of them suffices. Union-bound them and average over the conditioning. By Proposition~\ref{prop:chaininv}, before the first decoding failure the receiver follows the true walk, so any stream failure entails one of these events. No independence among different coefficient events is needed.
\end{proof}

\begin{inwords}
Either some key in the family came out unusually long, or none did. The first case has a probability that can be written down exactly, because $S$ is a sum of $2kn$ independent squares taking only the values $0,1,4$; the second reduces to Lemma~\ref{lem:fixedkey} with a known bound on~$S$. Averaging over the key instead, as in \eqref{eq:bernstein}, is dominated by the rare long keys that this argument isolates and charges for separately.
\end{inwords}

This proof applies to public-key-dependent message choice; it does not permit message selection from the encryption coins, and it gives no worst-case guarantee for an arbitrary exceptional family. Conditioning in an analysis is not rejection sampling in setup: the key-generation distribution is unchanged. Increasing $L$ affects the fresh-noise term linearly, so the bound degrades only logarithmically with stream length. The round choice 2400, about 1.56 times $\E S$, is close to but not exactly the optimum. Compression or different sampling would require a new calculation.

\begin{table}[htbp]
\centering\small
\caption{Recomputed bounds from \eqref{eq:streambound}. Except for the first row, thresholds minimise the bound over the integer grid $2200\le S_0\le2500$. Displayed decimals approximate an upper bound, not an observed failure rate.}
\label{tab:bounds}
\begin{tabular}{@{}rrrl@{}}
\toprule
$L$ & chosen $S_0$ & $\log_2$ failure upper bound & message choice\\
\midrule
64 & 2400 & $-192.5614$ & fixed or from public key\\
64 & 2402 & $-192.6152$ & fixed or from public key\\
1024 & 2394 & $-189.3077$ & fixed or from public key\\
32768 & 2384 & $-185.1846$ & fixed or from public key\\
\bottomrule
\end{tabular}
\end{table}

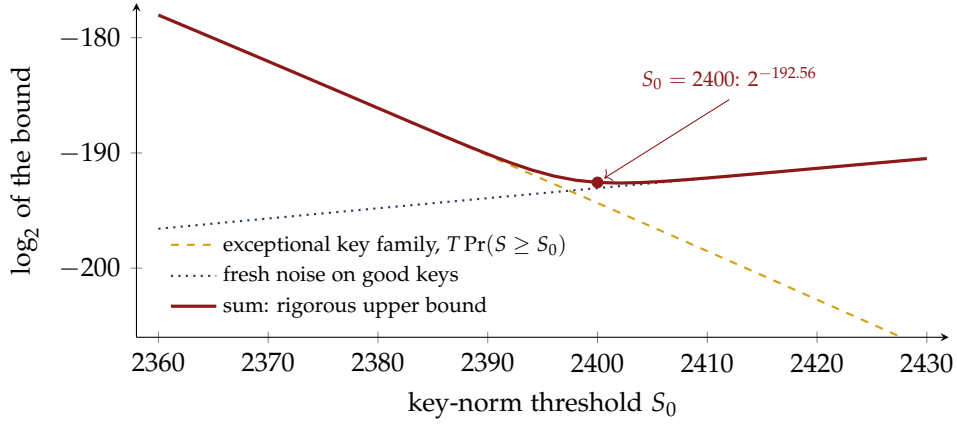
\begin{figure}[t]
\centering
\begin{tikzpicture}
\begin{axis}[width=0.78\textwidth,height=6cm,xmin=2358,xmax=2432,ymin=-206,ymax=-177,
  xlabel={key-norm threshold $S_0$},ylabel={$\log_2$ of the bound},axis lines=left,
  tick label style={font=\small,/pgf/number format/1000 sep={}},label style={font=\small},
  legend style={font=\scriptsize,draw=none,fill=none,at={(0.03,0.03)},anchor=south west,cells={anchor=west}}]
\addplot[zgold,thick,dashed] coordinates {(2360,-178.0339) (2362,-178.8347) (2364,-179.6371) (2366,-180.4410) (2368,-181.2465) (2370,-182.0536) (2372,-182.8622) (2374,-183.6724) (2376,-184.4842) (2378,-185.2975) (2380,-186.1124) (2382,-186.9288) (2384,-187.7468) (2386,-188.5663) (2388,-189.3874) (2390,-190.2101) (2392,-191.0343) (2394,-191.8600) (2396,-192.6873) (2398,-193.5162) (2400,-194.3466) (2402,-195.1785) (2404,-196.0120) (2406,-196.8470) (2408,-197.6836) (2410,-198.5218) (2412,-199.3614) (2414,-200.2026) (2416,-201.0454) (2418,-201.8897) (2420,-202.7355) (2422,-203.5829) (2424,-204.4318) (2426,-205.2822) (2428,-206.1342) (2430,-206.9877)};\addlegendentry{exceptional key family, $T\Pr(S\ge S_0)$}
\addplot[zblue,thick,dotted] coordinates {(2360,-196.5822) (2362,-196.4031) (2364,-196.2242) (2366,-196.0457) (2368,-195.8674) (2370,-195.6895) (2372,-195.5118) (2374,-195.3345) (2376,-195.1574) (2378,-194.9807) (2380,-194.8042) (2382,-194.6281) (2384,-194.4522) (2386,-194.2766) (2388,-194.1014) (2390,-193.9264) (2392,-193.7517) (2394,-193.5773) (2396,-193.4032) (2398,-193.2294) (2400,-193.0559) (2402,-192.8826) (2404,-192.7097) (2406,-192.5370) (2408,-192.3646) (2410,-192.1926) (2412,-192.0208) (2414,-191.8492) (2416,-191.6780) (2418,-191.5071) (2420,-191.3364) (2422,-191.1660) (2424,-190.9959) (2426,-190.8261) (2428,-190.6565) (2430,-190.4873)};\addlegendentry{fresh noise on good keys}
\addplot[zred,very thick] coordinates {(2360,-178.0339) (2362,-178.8347) (2364,-179.6371) (2366,-180.4410) (2368,-181.2465) (2370,-182.0535) (2372,-182.8620) (2374,-183.6720) (2376,-184.4833) (2378,-185.2957) (2380,-186.1089) (2382,-186.9219) (2384,-187.7330) (2386,-188.5390) (2388,-189.3335) (2390,-190.1043) (2392,-190.8301) (2394,-191.4769) (2396,-192.0013) (2398,-192.3657) (2400,-192.5614) (2402,-192.6152) (2404,-192.5704) (2406,-192.4661) (2408,-192.3289) (2410,-192.1747) (2412,-192.0119) (2414,-191.8448) (2416,-191.6758) (2418,-191.5060) (2420,-191.3359) (2422,-191.1657) (2424,-190.9958) (2426,-190.8260) (2428,-190.6565) (2430,-190.4873)};\addlegendentry{sum: rigorous upper bound}
\addplot[zred,only marks,mark=*,mark size=2pt] coordinates {(2400,-192.5614)};
\node[font=\scriptsize,zred] (s) at (axis cs:2412,-183.5) {$S_0=2400$: $2^{-192.56}$};
\draw[->,zred] (s.south) -- (axis cs:2400.6,-192.2);
\end{axis}
\end{tikzpicture}
\caption{The two terms of \eqref{eq:streambound} for $T=16$, $L=64$. The key term is computed from exact integer coefficients of \eqref{eq:Slaw}; the noise term uses $S+1\le S_0$ on the good-family event. Both fixed and public-key-chosen messages obey this bound.}
\label{fig:terms}
\end{figure}

\Needspace{7\baselineskip}
\section{Secret recovery as a nearby lattice point}\label{sec:lattice}
The public relation $b_t=As_t+e_t$ hides a small secret behind a small error. A geometric formulation of recovering $\sigma$ from $\beta$ should preserve both constraints.

\subsection{An augmented public lattice}
Put $D=kn$. Expand polynomials in the coefficient basis and let $\widetilde A\in\ZZ^{D\times D}$ be an integer lift of the negacyclic coefficient matrix of $A$. For a lift $\bar b_t$ of $b_t$, define
\begin{equation}\label{eq:auglattice}
\Lambda_A=\{(z,\widetilde Az+qv):z,v\in\ZZ^D\}\subset\RR^{2D},\qquad y_t=(0,\bar b_t).
\end{equation}
A column basis and the determinant are
\[
B_A=\begin{pmatrix}I_D&0\\ \widetilde A&qI_D\end{pmatrix},\qquad\det\Lambda_A=q^D .
\]
A different integer lift of $A$ gives the same lattice.

\begin{proposition}[Planted nearby point]\label{prop:planted}
The point $y_t^\star=(\bar s_t,\bar b_t-\bar e_t)$ belongs to $\Lambda_A$ and
\begin{equation}\label{eq:planted}
\|y_t-y_t^\star\|_2^2=\|\bar s_t\|_2^2+\|\bar e_t\|_2^2=S_t .
\end{equation}
If $\sqrt{S_t}<\lambda_1(\Lambda_A)/2$, then $y_t^\star$ is the unique closest lattice point to $y_t$.
\end{proposition}
\begin{proof}
The key relation gives $\bar b_t-\bar e_t=\widetilde A\bar s_t+qv$ for some integer $v$. This proves membership and the difference vector $(-\bar s_t,\bar e_t)$. For any other $y'\in\Lambda_A$, the triangle inequality gives $\|y_t-y'\|_2\ge\lambda_1(\Lambda_A)-\sqrt{S_t}>\sqrt{S_t}$.
\end{proof}

The associated minimisation problem is
\begin{equation}\label{eq:minprob}
\min_{z,v\in\ZZ^D}\Bigl(\|z\|_2^2+\|\bar b_t-\widetilde Az-qv\|_2^2\Bigr).
\end{equation}
In dimension $2D=1536$ the Gaussian heuristic gives
\[
\lambda_1(\Lambda_A)\approx\sqrt{\frac{2D}{2\pi e}}\,q^{1/2}\approx547.16,\qquad\sqrt{\E S_t}\approx39.19,
\]
roughly a factor of seven of slack in the uniqueness condition for typical instances. These are heuristic scale comparisons, not established shortest-vector lengths or attack costs for these structured lattices; the uniqueness conclusion requires its actual inequality, not the heuristic substitution.

\subsection{Why the secret norm cannot be discarded}
If $A$ is invertible modulo $q$, choosing $z=A^{-1}b_t$ makes the modular residual zero, although $z$ is generally not the sampled small secret: $A^{-1}b_t=s_t+A^{-1}e_t$, and smallness need not survive modular inversion. Likewise $\widetilde A\ZZ^D+q\ZZ^D=\ZZ^D$ in that case, so the unaugmented image lattice already contains $\bar b_t$. The augmented lattice avoids this degeneracy by retaining $\|z\|_2$.

The search problem is small-secret Module-LWE with the stipulated distributions~\cite{regev,lpr,ls}. Euclidean minimisation is not automatically the exact maximum-likelihood rule for CBD samples. Proposition~\ref{prop:planted} therefore supplies a precise geometric interpretation and a conditional uniqueness statement, not an unconditional equivalence between secret recovery and unique closest-vector decoding. Legitimate decryption cancels terms; it solves neither this optimisation nor an inverse problem for~$A$.

\Needspace{7\baselineskip}
\section{Assumptions and security of the chained composition}\label{sec:security}
The arithmetic follows the small-secret Module-LWE encryption pattern of Lindner--Peikert and Kyber~\cite{lp,kyber-spec,kyber-eurosp}, with an explicitly uniform matrix and independent coefficientwise sampling. The bit representative $\Delta$, the absence of compression, the key family and the framing are specified here; this is not a conforming implementation of ML-KEM~\cite{fips203}.

\begin{definition}[Decisional Module-LWE]\label{def:mlwe}
For $m,k\ge1$ let $\varepsilon_{m,k}(\tau)$ be the supremum of the absolute differences in acceptance probabilities over all time-$\tau$ distinguishers on $(M,Mr+h)$ and $(M,z)$, where $M\leftarrow\Rq^{m\times k}$ and $z\leftarrow\Rq^m$ are uniform and $r\in\Rq^k$, $h\in\Rq^m$ have independent coefficientwise $\CBD_\eta$ samples. Write
\[
\varepsilon_1=\varepsilon_{k,k},\qquad\varepsilon_2=\varepsilon_{k+1,k},\qquad\varepsilon_3=\varepsilon_{k+T,k},
\]
each evaluated at the adversary's time plus the polynomial simulation overhead.
\end{definition}

These are explicit distributional assumptions; this manuscript gives no concrete bound for any of them. A distinguisher for $m$ samples is also one for $m'\ge m$ samples, so $\varepsilon_2\le\varepsilon_3$; the point of $\varepsilon_3$ is that it is paid once per block rather than $T$ times.

\begin{lemma}[Shared-matrix hybrid]\label{lem:hybrid}
The distributions $(A,(As_t+e_t)_{t<T})$ and $(A,(z_t)_{t<T})$, with independent uniform $z_t$, differ computationally by at most $T\varepsilon_1$.
\end{lemma}
\begin{proof}
Replace one public vector at a time. At the challenged position use the supplied Module-LWE vector; generate earlier vectors uniformly and later vectors from fresh small secrets and errors with the supplied matrix. Summing the $T$ adjacent hybrid distances proves the claim.
\end{proof}

This is a statement about public-key distributions. It provides neither a $T$-fold attack-cost lower bound nor protection when a challenged secret is disclosed.

\subsection{Experiment, leakage and advantage conventions}
The adversary first receives $pk$ and outputs a classical message together with retained state; the message may depend on the public key. The challenger samples the nonce and any leading random plaintext blocks, then encrypts. In the ideal experiment it keeps the honestly generated public key but returns the same public header and length with a fresh uniform nonce and $L$ independent uniform pairs in $\Rq^k\times\Rq$.

Write $\mathrm{Adv}^{\mathrm{pr}}$ for the absolute difference of acceptance probabilities in the real and ideal experiments. For IND-CPA, the first stage outputs two messages of equal framed length and the same $\nu$; write $\mathrm{Adv}^{\mathrm{cpa}}$ for the absolute difference of acceptance probabilities under encryption of message~0 or message~1. Under the convention of success probability minus $1/2$, divide the CPA bound by two. The game provides no decryption oracle; ordinary encryption queries can be simulated locally.

The ideal ciphertext is uniform in ring elements, not in arbitrary bytes: coefficients are restricted to $\{0,\dots,3328\}$ and the tag and header have fixed syntax. Parameters, block count and framing granularity are public leakage. Random leading plaintext blocks are challenger coins, not public side information.

\begin{theorem}[Composition]\label{thm:composition}
For polynomial $T$, $L$, independent keys and encryption coins, and always-terminating, efficiently computable public state, selector and mask functions,
\begin{align}
\mathrm{Adv}^{\mathrm{pr}}&\le2T\varepsilon_1+L\min(T\varepsilon_2,\varepsilon_3),\label{eq:advpr}\\
\mathrm{Adv}^{\mathrm{cpa}}&\le2T\varepsilon_1+2L\min(T\varepsilon_2,\varepsilon_3).\label{eq:advcpa}
\end{align}
A maximum allowed block count may replace $L$. These confidentiality statements use no random-oracle assumption about~$H$.
\end{theorem}
\begin{proof}
First replace the public family by uniform vectors using Lemma~\ref{lem:hybrid}, at cost $T\varepsilon_1$, before the adversary chooses its message. In the resulting uniform-key game let $G_i$ have its first $i$ ciphertext pairs uniform and the rest honestly encrypted. All states, selectors and masks are computed from the actual plaintext and nonce, even when earlier ciphertext pairs have been replaced; they do not depend on those pairs. It remains to bound $|\Pr[G_i]-\Pr[G_{i+1}]|$ in two ways.

\emph{Route 1: guessing the fibre.} Guess $J\leftarrow[T]$ before publishing the key. Given a challenge $(M,y)$ with
\[
M=\begin{pmatrix}A^\top\\ z^\top\end{pmatrix}\in\Rq^{(k+1)\times k},
\]
place $z$ at public position $J$ and sample all other public vectors uniformly. Run the adversary's first stage, sample the nonce and any leading blocks, and compute $t_i,\mu_i$ from the known plaintext. If $t_i\neq J$, output a fixed bit independent of $y$. Otherwise put $u_i=y_{1..k}$ and $v_i=y_{k+1}+\Delta\mu_i$, simulate the first $i$ pairs uniformly and generate later pairs honestly. If $y=Mr+(f,g)$ this is an honest position-$i$ pair; if $y$ is uniform it is a uniform pair. All other encryptions, including later uses of $J$, need only public data. The published key has the same uniform distribution for every $J$, and before $y$ is used $J$ is independent of the key, the retained state, the chosen plaintext, the nonce and the padding. Averaging over $J$ therefore selects exactly the branch $J=t_i$, the signed acceptance-probability difference of the reduction is $1/T$ times that of the two adjacent games, and the unsuccessful branches have identical outputs. Each block transition costs $T\varepsilon_2$.

\emph{Route 2: embedding every fibre at once.} Take instead a challenge with
\[
M=\begin{pmatrix}A^\top\\ z_0^\top\\ \vdots\\ z_{T-1}^\top\end{pmatrix}\in\Rq^{(k+T)\times k},\qquad y=Mr+(f,g_0,\dots,g_{T-1})\ \text{or uniform},
\]
and publish $b_t=z_t$ for every $t$, which is exactly the uniform-key distribution. After the adversary has chosen its message, compute $t_i$ and put $u_i=y_{1..k}$, $v_i=y_{k+1+t_i}+\Delta\mu_i$; the rows of $y$ belonging to other fibres are discarded. In the real case $(u_i,v_i)=(A^\top r+f,\ \ip{b_{t_i}}r+g_{t_i}+\Delta\mu_i)$ is an honest pair with fresh coins; in the uniform case it is a uniform pair. No guess is made, and each block transition costs $\varepsilon_3$.

Summing over the $L$ positions costs $L\min(T\varepsilon_2,\varepsilon_3)$. For pseudorandomness, restore the honest public-key distribution once every pair is uniform; this costs another $T\varepsilon_1$ and gives \eqref{eq:advpr}. The ideal pair distribution then depends on the selected message only through its public length. For CPA, go from encryption of message~0 to the uniform-key/uniform-pair middle game and reverse the argument for message~1; equal framed lengths make that middle distribution identical. There are two public-key transitions and two sets of block transitions, which yields \eqref{eq:advcpa}.
\end{proof}

\begin{inwords}
The reduction replaces the ciphertext by noise one block at a time. To replace block $i$ it must plant its Module-LWE challenge at the fibre that block actually reads, but the walk is chosen by the adversary after seeing the key. Route~1 guesses the fibre and is right with probability $1/T$; Route~2 plants a challenge in every fibre at once, using a Module-LWE instance with $T$ more rows. Either way the theorem is conditional confidentiality for the complete stream, not a concrete security-level estimate.
\end{inwords}

\subsection{Interpretation and limits of the reduction}
For $T=16$, $L=64$ the CPA bound is $32\varepsilon_1+2048\varepsilon_2$ by Route~1, multiplicative losses of 5 and 11 bits, or $32\varepsilon_1+128\varepsilon_3$ by Route~2, a loss of 7 bits on an assumption with $k+T=19$ rows. At the maximum permitted length $L=32768$ the block-transition factors $2LT$ and $2L$ become $2^{20}$ and $2^{16}$ respectively. Assigning a concrete security level would require concrete advantage bounds at the reduction's running time and an allocation of the total error budget across the terms; the losses alone are not such an estimate.

If the complete framed plaintext is independent of the public key, the reduction can determine the target index before embedding its challenge row, and the pseudorandomness bound improves to $2T\varepsilon_1+L\varepsilon_2$. This fixed-message observation alone does not prove standard IND-CPA. Conversely, allowing messages chosen after $pk$ in Theorem~\ref{thm:composition} is ordinary CPA adaptivity: it does not establish security for encryptions of functions of an unknown secret key, as required in KDM games.

Known plaintext reveals states and future selectors through \eqref{eq:chain}, and the first index is always public, so absolute hiding of the walk is false. The theorem says that ciphertexts do not efficiently distinguish equal-length candidate messages under the assumptions and with no secret exposure; it does not make their computable candidate walks secret.

The reductions are straight-line. They also apply to quantum distinguishers if the assumptions bound quantum algorithms, with classical public keys, chosen messages and challenge ciphertexts, and retained quantum state passed through without rewinding. Superposition encryption access is not modelled. No authentication or IND-CCA security follows; a Fujisaki--Okamoto-style successor needs a separate construction and proof~\cite{fo,hhk}.

\subsection{What the mask does, and what must not be inferred}
The confidentiality reduction works even with a constant mask, since it computes the real mask from known plaintext. The mask matters instead in the exposure discussion: if an exposed key decodes a later block, it obtains $\mu_i=m_i\oplus\kappa_i$, whereas without the mask it would obtain $m_i$ directly.

Unpredictability of a public hash on an unknown state does not, however, follow from marginal state entropy alone, especially with auxiliary information and adversarially chosen states. A random-oracle analysis would have to bound queries hitting the hidden state and condition its entropy on the adversary's entire view; a keyed-PRF claim would additionally require an appropriate unavailable key, and the public \spec{} mask is not such a keyed construction. We do not prove an unrestricted masking or partial-exposure theorem. The direct-prefix calculation below is explicitly a procedural model.

\Needspace{7\baselineskip}
\section{Exact behaviour under partial exposure}\label{sec:exposure}
\begin{definition}[Direct-prefix model]\label{def:prefix}
Fix an exposed set $K\subseteq[T]$ independently of the selector trajectory, and put $p=|K|/T$. Assume selectors are independent and uniform, and ignore decoding errors. A direct decryptor advances at a position exactly when its selected index belongs to $K$; at the first missing key it stops. It does not guess plaintext or states, use external plaintext knowledge, exploit leakage or attack the primitive. Let $N$ be the number of recovered blocks, capped at~$L$.
\end{definition}

For this procedure the recovered set is always an initial segment: the decryptor needs $x_{i-1}$ to compute $t_i$ and $\kappa_i$, and $x_i$ cannot be formed without $m_i$. On the bundle, the walk can be followed only while it stays over the exposed part $K$ of the base. This follows from the stopping rule and requires no hash idealisation. It must not be reformulated as ``the state cannot be computed for every hash'': a constant or predictable state function is a counterexample, and the cryptographic difficulty of recovering an unknown next state is a separate question.

\begin{theorem}[Capped geometric prefix]\label{thm:prefix}
Under Definition~\ref{def:prefix},
\begin{align}
\Pr(N=j)&=p^j(1-p),&&0\le j<L,\label{eq:prefixlaw}\\
\Pr(N=L)&=p^L,&&\label{eq:prefixfull}\\
\E N&=\sum_{j=1}^Lp^j=\frac{p(1-p^L)}{1-p},&&0\le p<1.\label{eq:prefixmean}
\end{align}
For $p=1$, $N=L$.
\end{theorem}
\begin{proof}
Recovering exactly $j<L$ blocks requires $j$ consecutive selections in $K$ followed by one outside it; full recovery requires $L$ hits. The tail-sum identity gives $\E N=\sum_{j=1}^L\Pr(N\ge j)=\sum_{j=1}^Lp^j$.
\end{proof}

At $T=16$, $|K|=8$, $L=64$ the expected prefix is $1-2^{-64}$ blocks and full direct recovery has probability $2^{-64}$. Independently key-labelled blocks without a prefix-dependent mask instead permit $pL=32$ blocks on average in the analogous model. The comparison does not grant 64 additional security bits or bound an unrestricted adversary.

Real SHAKE trajectories are deterministic. Distinct fresh random-oracle inputs can motivate the independent-selector model, but input collisions, adaptive observations and known plaintext must be treated in any full analysis. Independence of the sampled secrets also does not create operational compartmentalisation: separate storage or access controls are required for meaningful per-fibre exposure.

\begin{figure}[t]
\centering
\begin{tikzpicture}
\begin{axis}[name=L,width=0.47\textwidth,height=5.4cm,xmin=-0.6,xmax=12.6,ymin=-0.02,ymax=0.8,
  xlabel={directly recovered prefix length $j$},ylabel={$\Pr[N=j]$},axis lines=left,
  tick label style={font=\scriptsize},label style={font=\small},
  legend style={font=\scriptsize,draw=none,at={(0.97,0.97)},anchor=north east}]
\addplot[zblue,thick,mark=*,mark size=1.5pt] coordinates {(0,0.75000) (1,0.18750) (2,0.04688) (3,0.01172) (4,0.00293) (5,0.00073) (6,0.00018) (7,0.00005) (8,0.00001) (9,0.00000) (10,0.00000) (11,0.00000) (12,0.00000)};\addlegendentry{$p=0.25$}
\addplot[zred,thick,mark=square*,mark size=1.5pt] coordinates {(0,0.50000) (1,0.25000) (2,0.12500) (3,0.06250) (4,0.03125) (5,0.01562) (6,0.00781) (7,0.00391) (8,0.00195) (9,0.00098) (10,0.00049) (11,0.00024) (12,0.00012)};\addlegendentry{$p=0.5$}
\addplot[zgold,thick,mark=triangle*,mark size=1.8pt] coordinates {(0,0.25000) (1,0.18750) (2,0.14062) (3,0.10547) (4,0.07910) (5,0.05933) (6,0.04449) (7,0.03337) (8,0.02503) (9,0.01877) (10,0.01408) (11,0.01056) (12,0.00792)};\addlegendentry{$p=0.75$}
\end{axis}
\begin{axis}[at={(L.east)},anchor=west,xshift=1.9cm,width=0.47\textwidth,height=5.4cm,ymode=log,xmin=0.5,xmax=15.5,ymin=0.04,ymax=150,
  xtick={1,4,8,12,15},xlabel={exposed keys out of 16},ylabel={expected recovered blocks},axis lines=left,
  tick label style={font=\scriptsize},label style={font=\small},
  legend style={font=\scriptsize,draw=none,at={(0.5,1.04)},anchor=south}]
\addplot[zred,thick,mark=*,mark size=1.5pt] coordinates {(1,0.0666667) (2,0.142857) (3,0.230769) (4,0.333333) (5,0.454545) (6,0.6) (7,0.777778) (8,1) (9,1.28571) (10,1.66667) (11,2.2) (12,3) (13,4.33333) (14,6.99864) (15,14.7589)};\addlegendentry{direct-prefix model}
\addplot[zgold,thick,dashed] coordinates {(1,4) (2,8) (3,12) (4,16) (5,20) (6,24) (7,28) (8,32) (9,36) (10,40) (11,44) (12,48) (13,52) (14,56) (15,60)};\addlegendentry{independent key-labelled blocks}
\draw[zblue] (axis cs:8,1) circle (3.5pt);
\node[font=\scriptsize,zblue] (o) at (axis cs:4.5,6) {1 block vs 32};
\draw[->,zblue] (o.south east) -- (axis cs:7.8,1.15);
\end{axis}
\end{tikzpicture}
\caption{Exact predictions of the direct-prefix model, $L=64$. Left: the prefix distribution \eqref{eq:prefixlaw}. Right: expected recovered blocks \eqref{eq:prefixmean} compared with independently key-labelled blocks; $|K|=16$ is omitted, since both equal $L=64$. Neither panel is a measured attack-success curve.}
\label{fig:prefix}
\end{figure}
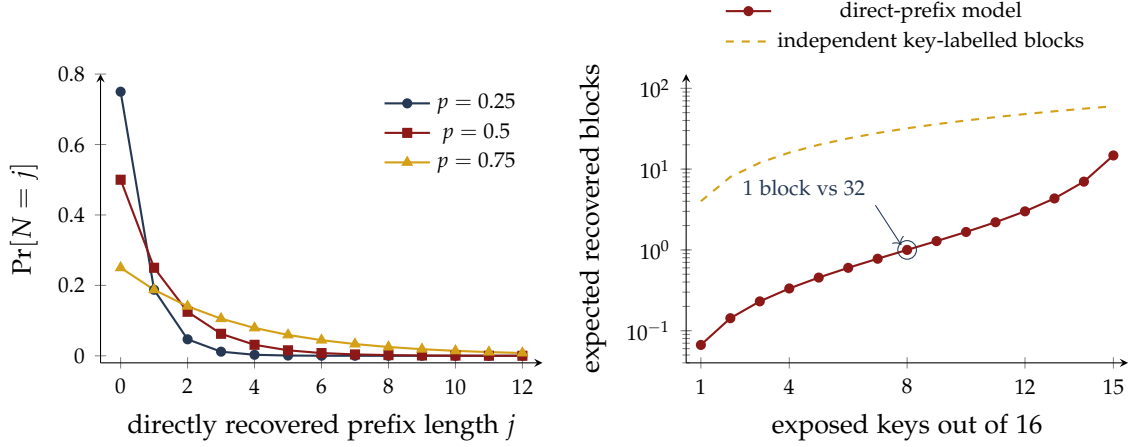

\subsection{Several streams and leading random blocks}
Assume additionally that the $W$ stream trajectories are independent. For $1\le j\le L$,
\begin{equation}\label{eq:multistream}
\Pr\Bigl(\max_{a\le W}N_a\ge j\Bigr)=1-(1-p^j)^W,\qquad\E\max_{a\le W}N_a=\sum_{j=1}^L\bigl[1-(1-p^j)^W\bigr].
\end{equation}
With $\nu$ leading random blocks and $\ell$ framed-message blocks, let $R_a$ count the framed-message blocks directly recovered from stream $a$. Then
\begin{align}
\Pr(R_a\ge j)&=p^{\nu+j},&&1\le j\le\ell,\label{eq:padtail}\\
\E R_a&=p^\nu\,\frac{p(1-p^\ell)}{1-p},&&0\le p<1,\label{eq:padmean}\\
\E\max_{a\le W}R_a&=\sum_{j=1}^\ell\bigl[1-(1-p^{\nu+j})^W\bigr].\label{eq:padmax}
\end{align}
For $p=0$ recovery is zero; for $p=1$ it is $\ell$. These identities follow from consecutive hits, independence across streams and the tail-sum formula. Different real nonces do not by themselves prove the required independence. Without cross-stream independence, the union bound still yields $\Pr(\max_aR_a\ge j)\le\min(1,Wp^{\nu+j})$ when the marginal law holds. Without padding and for $0<p<1$, writing $a=\lceil\log W/\log(1/p)\rceil$ also gives $\E\max N_a\le\min\bigl(L,\,a+p/(1-p)\bigr)$ by splitting the tail sum.

At $p=1/2$, $\nu=8$, $\ell=64$, one stream has expected framed-block recovery $2^{-8}(1-2^{-64})$ and probability $2^{-9}$ of any such recovery. Across $W=256$ independent streams the probability becomes
\[
1-(1-2^{-9})^{256}=0.3937658\ldots
\]
The padding costs $8\cdot1536=12288$ ciphertext bytes. The sufficient model-level condition $Wp^{\nu+1}\le\alpha$ bounds the probability of any framed-block recovery by $\alpha$. A framed block includes metadata, so $R_a$ is not a count of recovered user bytes. Known blocks can be inserted directly into the state recurrence, so predictable zero padding cannot provide the intended barrier, and random padding does not remove guessing, side channels or candidate-message attacks.

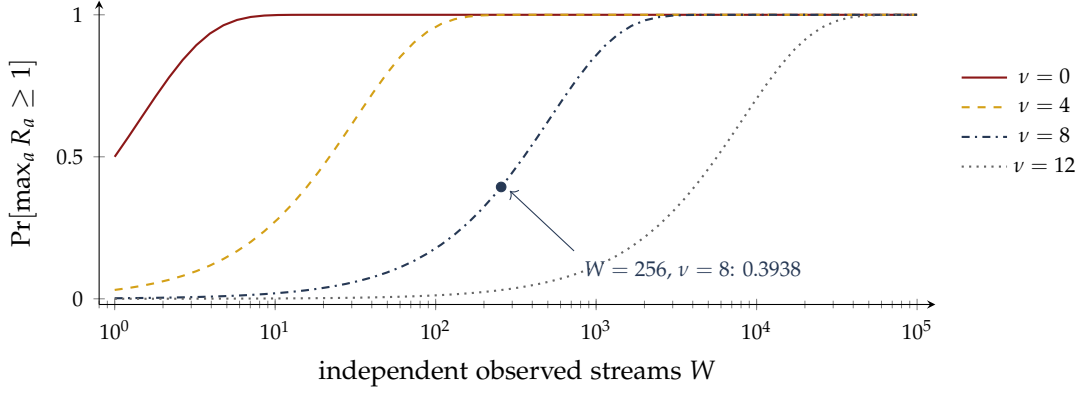
\begin{figure}[t]
\centering
\begin{tikzpicture}
\begin{axis}[width=0.8\textwidth,height=5.6cm,xmode=log,xmin=0.8,xmax=1.3e5,ymin=-0.02,ymax=1.05,
  xlabel={independent observed streams $W$},ylabel={$\Pr[\max_a R_a\ge1]$},axis lines=left,
  tick label style={font=\scriptsize},label style={font=\small},
  legend style={font=\scriptsize,draw=none,at={(1.02,0.6)},anchor=west}]
\addplot[zred,thick] coordinates {(1,0.50000) (1.2155,0.56937) (1.4774,0.64086) (1.7957,0.71197) (2.1826,0.77973) (2.6529,0.84101) (3.2246,0.89302) (3.9194,0.93391) (4.7639,0.96319) (5.7904,0.98193) (7.0381,0.99239) (8.5547,0.99734) (10.398,0.99926) (12.638,0.99984) (15.362,0.99998) (18.672,1.00000) (22.695,1.00000) (27.585,1.00000) (33.529,1.00000) (40.754,1.00000) (49.535,1.00000) (60.209,1.00000) (73.182,1.00000) (88.951,1.00000) (108.12,1.00000) (131.41,1.00000) (159.73,1.00000) (194.15,1.00000) (235.98,1.00000) (286.83,1.00000) (348.64,1.00000) (423.76,1.00000) (515.07,1.00000) (626.05,1.00000) (760.95,1.00000) (924.91,1.00000) (1124.2,1.00000) (1366.4,1.00000) (1660.9,1.00000) (2018.8,1.00000) (2453.8,1.00000) (2982.5,1.00000) (3625.1,1.00000) (4406.2,1.00000) (5355.7,1.00000) (6509.7,1.00000) (7912.3,1.00000) (9617.2,1.00000) (11690,1.00000) (14208,1.00000) (17270,1.00000) (20991,1.00000) (25514,1.00000) (31012,1.00000) (37694,1.00000) (45816,1.00000) (55688,1.00000) (67688,1.00000) (82272,1.00000) (1e+05,1.00000)};\addlegendentry{$\nu=0$}
\addplot[zgold,thick,dashed] coordinates {(1,0.03125) (1.2155,0.03785) (1.4774,0.04582) (1.7957,0.05542) (2.1826,0.06695) (2.6529,0.08078) (3.2246,0.09731) (3.9194,0.11701) (4.7639,0.14037) (5.7904,0.16793) (7.0381,0.20025) (8.5547,0.23784) (10.398,0.28116) (12.638,0.33052) (15.362,0.38597) (18.672,0.44723) (22.695,0.51351) (27.585,0.58347) (33.529,0.65510) (40.754,0.72580) (49.535,0.79251) (60.209,0.85215) (73.182,0.90206) (88.951,0.94064) (108.12,0.96770) (131.41,0.98458) (159.73,0.99373) (194.15,0.99790) (235.98,0.99944) (286.83,0.99989) (348.64,0.99998) (423.76,1.00000) (515.07,1.00000) (626.05,1.00000) (760.95,1.00000) (924.91,1.00000) (1124.2,1.00000) (1366.4,1.00000) (1660.9,1.00000) (2018.8,1.00000) (2453.8,1.00000) (2982.5,1.00000) (3625.1,1.00000) (4406.2,1.00000) (5355.7,1.00000) (6509.7,1.00000) (7912.3,1.00000) (9617.2,1.00000) (11690,1.00000) (14208,1.00000) (17270,1.00000) (20991,1.00000) (25514,1.00000) (31012,1.00000) (37694,1.00000) (45816,1.00000) (55688,1.00000) (67688,1.00000) (82272,1.00000) (1e+05,1.00000)};\addlegendentry{$\nu=4$}
\addplot[zblue,thick,dashdotted] coordinates {(1,0.00195) (1.2155,0.00237) (1.4774,0.00288) (1.7957,0.00350) (2.1826,0.00426) (2.6529,0.00517) (3.2246,0.00628) (3.9194,0.00763) (4.7639,0.00927) (5.7904,0.01126) (7.0381,0.01367) (8.5547,0.01659) (10.398,0.02012) (12.638,0.02441) (15.362,0.02959) (18.672,0.03585) (22.695,0.04340) (27.585,0.05250) (33.529,0.06345) (40.754,0.07658) (49.535,0.09230) (60.209,0.11105) (73.182,0.13331) (88.951,0.15962) (108.12,0.19053) (131.41,0.22657) (159.73,0.26822) (194.15,0.31584) (235.98,0.36957) (286.83,0.42923) (348.64,0.49419) (423.76,0.56328) (515.07,0.63468) (626.05,0.70593) (760.95,0.77410) (924.91,0.83606) (1124.2,0.88896) (1366.4,0.93085) (1660.9,0.96111) (2018.8,0.98068) (2453.8,0.99175) (2982.5,0.99706) (3625.1,0.99916) (4406.2,0.99982) (5355.7,0.99997) (6509.7,1.00000) (7912.3,1.00000) (9617.2,1.00000) (11690,1.00000) (14208,1.00000) (17270,1.00000) (20991,1.00000) (25514,1.00000) (31012,1.00000) (37694,1.00000) (45816,1.00000) (55688,1.00000) (67688,1.00000) (82272,1.00000) (1e+05,1.00000)};\addlegendentry{$\nu=8$}
\addplot[black!60,thick,dotted] coordinates {(1,0.00012) (1.2155,0.00015) (1.4774,0.00018) (1.7957,0.00022) (2.1826,0.00027) (2.6529,0.00032) (3.2246,0.00039) (3.9194,0.00048) (4.7639,0.00058) (5.7904,0.00071) (7.0381,0.00086) (8.5547,0.00104) (10.398,0.00127) (12.638,0.00154) (15.362,0.00187) (18.672,0.00228) (22.695,0.00277) (27.585,0.00336) (33.529,0.00408) (40.754,0.00496) (49.535,0.00603) (60.209,0.00732) (73.182,0.00889) (88.951,0.01080) (108.12,0.01311) (131.41,0.01591) (159.73,0.01931) (194.15,0.02342) (235.98,0.02840) (286.83,0.03441) (348.64,0.04167) (423.76,0.05042) (515.07,0.06094) (626.05,0.07358) (760.95,0.08871) (924.91,0.10677) (1124.2,0.12824) (1366.4,0.15364) (1660.9,0.18352) (2018.8,0.21843) (2453.8,0.25885) (2982.5,0.30517) (3625.1,0.35760) (4406.2,0.41603) (5355.7,0.47994) (6509.7,0.54828) (7912.3,0.61937) (9617.2,0.69089) (11690,0.75998) (14208,0.82351) (17270,0.87855) (20991,0.92289) (25514,0.95561) (31012,0.97731) (37694,0.98996) (45816,0.99628) (55688,0.99888) (67688,0.99974) (82272,0.99996) (1e+05,1.00000)};\addlegendentry{$\nu=12$}
\addplot[zblue,only marks,mark=*,mark size=1.8pt] coordinates {(256,0.3938)};
\node[font=\scriptsize,zblue] (w) at (axis cs:4000,0.1) {$W=256$, $\nu=8$: 0.3938};
\draw[->,zblue] (w.north west) -- (axis cs:290,0.38);
\end{axis}
\end{tikzpicture}
\caption{Probability $\Pr[\max_{a\le W}R_a\ge1]$ of any framed-block recovery in the independent-stream direct-prefix model, $p=1/2$. Leading random blocks reduce the per-stream probability; repeated observations increase the chance that some stream has a long exposed prefix.}
\label{fig:streams}
\end{figure}

\subsection{A known-key membership test}
Given an exposed $s_t$ and an observed pair $(u,v)$, compute $w=v-\ip{s_t}u$. Count the coefficients within circular distance less than $q/8$ of either $0$ or $\Delta$, call this count $C_t$, and accept membership if $C_t\ge192=3n/4$.

\begin{proposition}[Fixed-index membership experiment]\label{prop:membership}
Fix indices $t,j$ independently of the sampled key family and consider one fresh block encrypted at $j$. Its bit polynomial may be chosen from the public key and the disclosed $s_t$ before the fresh block coins. If $j=t$, the false-negative probability is at most
\begin{equation}\label{eq:fn}
\rho_{\mathrm{fn}}=2n\exp\Bigl(-\frac{417^2}{2(1537+4\cdot417/3)}\Bigr)=4.66028\ldots\times10^{-16}.
\end{equation}
If $j\neq t$, the false-positive probability is at most $\rho_{\mathrm{fp}}+\varepsilon_1+\varepsilon_2$, where
\begin{equation}\label{eq:fp}
\rho_{\mathrm{fp}}=\sum_{a=192}^{256}\binom{256}{a}\Bigl(\frac{1666}{3329}\Bigr)^a\Bigl(\frac{1663}{3329}\Bigr)^{256-a}=2.75003\ldots\times10^{-16}.
\end{equation}
The distinguishing gap is therefore at least $1-\rho_{\mathrm{fn}}-\rho_{\mathrm{fp}}-\varepsilon_1-\varepsilon_2$ in this fixed-index experiment.
\end{proposition}
\begin{proof}
If $j=t$ and every $|\delta_c|\le416$, then $C_t=n$. Apply the averaged Bernstein bound at threshold 417 and union-bound over coefficients; this bounds even the larger event $C_t<n$, which contains every false negative. If $j\neq t$, replace $b_j$ by a uniform vector, at cost $\varepsilon_1$, while generating the disclosed $s_t$, its public vector and the other keys honestly. Next embed a $(k+1)\times k$ challenge with matrix rows $(A^\top,b_j^\top)$ to replace the fresh pair by a uniform pair, at cost $\varepsilon_2$; the known-secret row $b_t$ can still be generated honestly using this $A$. In the uniform game $w$ is uniform in $\Rq$. Each of the two radius-$q/8$ regions contains 833 residues and the regions are disjoint, so 1666 residues qualify, and independent uniform coefficients give the binomial tail \eqref{eq:fp}.
\end{proof}

This precise experiment avoids silently applying a fixed-index reduction to an adaptively selected target. For whole-stream or adaptively selected tests, hybrid losses, conditioning and accumulated test errors must be accounted for separately. An exposed secret nevertheless gives a concrete way to investigate positions beyond a directly decrypted prefix, and may yield their masked blocks.

\subsection{Candidate-message attack and the correct interpretation of its probability}
Suppose two fully known candidate block streams first differ at position $j$. From either candidate and the public nonce the attacker computes a candidate walk. Assume, in addition to Definition~\ref{def:prefix}, that the compared states remain distinct and that the two post-divergence selector sequences are independent uniform sequences, across candidates and positions; this is a stronger idealisation than the one-stream prefix model.

A position $i>j$ supplies a \emph{membership witness} if its candidate selectors differ and at least one is exposed. It supplies a \emph{decoding witness} if the candidate selectors coincide and lie in $K$: decoding with the exposed key yields $\mu_i$, while the candidates predict $m_i\oplus\kappa_i$ from their distinct states, and these predictions differ unless the two masks happen to compensate the plaintext difference, an event ignored in this idealisation. If tests were exact, either witness would distinguish the candidates. A position therefore supplies no witness exactly when both candidate selectors lie outside $K$, with probability
\begin{equation}\label{eq:q0}
q_0=(1-p)^2 .
\end{equation}
Counting membership witnesses alone would give the larger value $q_0^{\mathrm{mem}}=p/T+(1-p)^2$. For $m=L-1-j$ independent compared positions the probability of no witness is $q_0^m$. At $T=16$, $p=1/2$, $m=32$ this is
\[
q_0^{m}=2^{-64}=5.42101\ldots\times10^{-20},\qquad\text{versus}\qquad (q_0^{\mathrm{mem}})^{m}=(9/32)^{32}=2.34942\ldots\times10^{-18}.
\]
Position $j$ itself, where both candidates share the state and hence the selector, supplies a further decoding witness with probability $p$; it is omitted here, so $q_0^m$ upper-bounds the no-witness probability of the whole comparison.
This is not the full attack's error probability. With perfect tests and random guessing when there is no witness, success is $1-q_0^m/2$. If $\varphi_{\mathrm{test}}$ bounds the probability that any test used is wrong under the actual experiment, a valid lower bound is
\begin{equation}\label{eq:guess}
\Pr(\text{successful guess})\ge1-\tfrac12q_0^m-\varphi_{\mathrm{test}}.
\end{equation}
A complete estimate must justify $\varphi_{\mathrm{test}}$, the extra idealisation and any state-collision loss; Proposition~\ref{prop:membership} alone does not supply all these bounds for adaptive whole-stream tests. If the candidates include the leading random blocks, the experiment also differs from one in which that padding is hidden challenger randomness.

\begin{inwords}
The prefix law describes an attacker who stops at its first missing key. A more capable attacker may compute walks from candidate messages and use an exposed key to test later positions. These capabilities are why the prefix law must not be advertised as unrestricted partial-compromise security.
\end{inwords}

\Needspace{7\baselineskip}
\section{Serial dependence and parallel work}\label{sec:serial}
The prescribed decryption loop needs a recovered block to compute the state selecting the next block. This is a dependency in that algorithm. A party holding all secrets can instead precompute
\begin{equation}\label{eq:table}
d_{i,t}=\Decode\bigl(v_i-\ip{s_t}{u_i}\bigr),\qquad0\le i<L,\ t\in[T],
\end{equation}
independently for every position and key, reading the whole section at every position rather than along the walk, and then traverse the table using the state recurrence.

\begin{figure}[htbp]
\centering
\begin{tikzpicture}[box/.style={draw=zred,rounded corners=2pt,align=center,font=\small,inner sep=5pt,minimum height=1.1cm,text width=3.2cm},>=Stealth]
\node[box] (a) {all ciphertext pairs\\ and secret section};
\node[box,right=0.8cm of a] (b) {parallel candidate\\ table $d_{i,t}$ for all $(i,t)$};
\node[box,right=0.8cm of b] (c) {walk-guided traversal\\ $m_i=d_{i,t_i}\oplus\kappa_i$};
\draw[->,zred,thick] (a)--(b); \draw[->,zred,thick] (b)--(c);
\node[below=0.15cm of b,font=\scriptsize,black!70] {work $LTW_{\mathrm{dec}}$, storage $LTn$ bits};
\node[below=0.15cm of c,font=\scriptsize,black!70] {$L$ dependent updates};
\end{tikzpicture}
\caption{Where the remaining dependency sits in the candidate-table algorithm. This is an explicit upper-bound algorithm with a work--storage trade-off, not an adversarial lower bound.}
\label{fig:table}
\end{figure}
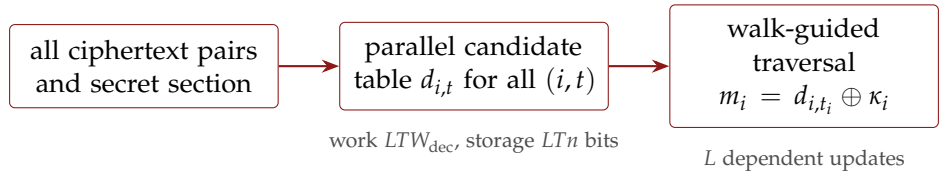

\par\medskip\noindent\begin{minipage}{\linewidth}
\begin{lstlisting}[caption={Candidate-table decryption.}]
Parallel stage:
  for every (i,t), independently:
      d[i,t] <- Decode(v[i] - inner(s[t],u[i]))
Sequential traversal:
  x <- H(INIT,iv)
  for i = 0,...,L-1:
      t <- H(FIBRE,x)
      m[i] <- d[i,t] (+) H(MASK,x)
      x <- H(CHAIN,i,x,m[i])
  return Unframe(m[nu:])
\end{lstlisting}
\end{minipage}\par\medskip

With block-decoding work $W_{\mathrm{dec}}$ and depth $D_{\mathrm{dec}}$, stage~1 has work $LTW_{\mathrm{dec}}$ and depth $D_{\mathrm{dec}}$ with sufficient processors. Its fully materialised table contains $LTn$ bits, excluding indexing overhead; one batch of decodings is not one layer of bit operations. With only an exposed set $K$, the analogous table has $L|K|$ entries. The second stage still performs selector, mask and state computations at each position. At $T=16$, $L=64$ there are 1024 candidate decodings, a 32768-byte decoded table and 64 dependent \lab{CHAIN} updates, with more than 64 hash calls in total.

Proving a universal lower bound requires a computational model allowing preprocessing, parallel queries, speculative branches, known plaintext and limits on total work. The construction is neither a proof of sequential work nor a verifiable delay function; the literature on sequential work supplies relevant models~\cite{posw}.

\subsection{A conditional divergence calculation}
Suppose changing block $j$ gives a different state and all later compared state queries remain fresh and distinct in an ideal model. For $m=L-1-j$ later selectors, the probability that every pair of indices still matches is $T^{-m}$. Averaging over a uniformly chosen changed position gives, for $T>1$,
\begin{equation}\label{eq:divergence}
\Pr(\text{no selected-index divergence})=\frac1L\sum_{m=0}^{L-1}T^{-m}=\frac{1-T^{-L}}{L(1-T^{-1})}.
\end{equation}
At $L=64$, $T=16$ this is $1/60$ up to a relative error $16^{-64}$. Changing the last block cannot change an already-used selector, and $T=1$ gives probability one. The calculation concerns whether some later index differs; it does not assert that every later index changes or prove an avalanche property.

\Needspace{7\baselineskip}
\section{Explicit-family storage and work}\label{sec:storage}
If every coefficient occupies $\lceil\log_2q\rceil$ bits and all objects are listed explicitly, the payload sizes are
\begin{equation}\label{eq:sizes}
|pk|=\lceil\log_2q\rceil\,n(k^2+kT),\qquad|sk|=\lceil\log_2q\rceil\,nkT,\qquad|(u,v)|=\lceil\log_2q\rceil\,n(k+1)
\end{equation}
bits. These count $k^2+kT$, $kT$ and $k+1$ polynomials respectively. At the default parameters a polynomial occupies 384 bytes, giving 21888 public-key bytes, 18432 private-key bytes and 1536 bytes per ciphertext pair. Small secrets admit more compact encodings; these are residue-format sizes, not entropy lower bounds.

For polynomial $n$, $k$, $\log q$, $L$, polynomial explicit public-key length implies polynomial $T$ and polynomial candidate-table work. This says nothing about whether that work is practical, or about compressed or procedurally generated families. Public seed expansion needs a specified pseudorandom generation model, and deriving keys from a master secret changes the exposure analysis because master-secret compromise exposes the derived family.

Correct decryption does not imply injectivity of $s\mapsto b$. With $A=I$, the small pairs $(s,e)=(0,a)$ and $(a,0)$ yield the same $b=a$, and both can satisfy an appropriate correctness margin. Computational public-key pseudorandomness must not be confused with preservation of the secret's conditional Shannon entropy after the public key is known.

\Needspace{7\baselineskip}
\section{Three simpler baselines}\label{sec:baselines}
\paragraph{Independent KEMs with symmetric encryption.} Generate $T$ independent ML-KEM pairs and choose one for each message, or combine encapsulations through a separately analysed KEM combiner~\cite{fips203,combiners}. Encrypt bulk data with a suitable authenticated-encryption scheme. The choice is ordinarily fixed per message and expansion is close to one for large payloads. A KEM's IND-CCA property does not authenticate the sender; the surrounding protocol must specify that requirement.

\paragraph{Threshold secret sharing.} Split a payload into $T$ shares with a threshold scheme~\cite{shamir} and protect each under its intended key. Fewer than the threshold number of plaintext shares reveal no payload information in the ideal sharing model; composition with encrypted shares and metadata needs its own statement. This can give a stronger threshold guarantee than a direct-prefix law, at greater share volume and with different access semantics.

\paragraph{All-or-nothing transforms.} An all-or-nothing transform followed by encryption can make recovery depend on possession of all transformed blocks~\cite{rivest,boyko}. Exactly what a missing block hides depends on the transform, the adversarial model, brute-force costs and the encryption composition, and such transforms generally sacrifice immediate incremental plaintext release.

\medskip
For Z-Sigil, the exact 2048-byte example is 99912 bytes, about $48.79\times$ the user payload after header and framing; the often-quoted $48\times$ is the per-block ratio $1536/32$. Similarity of the underlying Module-LWE arithmetic does not make the security assumptions and full KEM analyses identical. The comparison motivates a research question, not a deployment advantage: a standard KEM with a suitable symmetric protocol remains the practical baseline against which any proposed successor should be evaluated.

\begin{table}[t]
\centering\small
\caption{Explicit storage and protocol distinctions. Standard ML-KEM-768 sizes are from FIPS~203~\cite{fips203}. The last column selects one instance; an all-instance combiner would have different ciphertext costs. AEAD overhead and sender authentication depend on the protocol.}
\label{tab:compare}
\begin{tabularx}{\textwidth}{@{}>{\raggedright\arraybackslash}p{2.4cm}>{\raggedright\arraybackslash}X>{\raggedright\arraybackslash}X>{\raggedright\arraybackslash}X@{}}
\toprule
Property & Z-Sigil & ML-KEM-768 + AEAD & 16 independent ML-KEM-768 + AEAD\\
\midrule
Public key & 21\,888 B & 1\,184 B & 18\,944 B\\
Private key & 18\,432 B & 2\,400 B & 38\,400 B\\
2048-byte payload & $65\cdot1536+72$ B & 1\,088 B + AEAD record & 1\,088 B + AEAD record\\
Bulk expansion & $48\times$ asymptotically & near $1\times$ & near $1\times$\\
Confidentiality claim & conditional CPA & standardised KEM + protocol & per-instance KEM + protocol\\
Message integrity & absent & AEAD, with protocol requirements & AEAD, with protocol requirements\\
Key selection & plaintext-fed, per block & \multicolumn{2}{l}{one KEM chosen per message}\\
Partial exposure & direct-prefix model only & depends on the compromised key & messages using compromised keys\\
\bottomrule
\end{tabularx}
\end{table}

\Needspace{7\baselineskip}
\section{Geometry and a transport research programme}\label{sec:transport}
\subsection{What the current geometry establishes}
Two geometric objects appear in this paper and they answer different questions. The key bundle $\EP$ with its sections $\sigma,\epsilon,\beta$ describes how key material is organised and how a message reads it (Section~3.2); the augmented lattice $\Lambda_A$ describes the recovery of a small section from its noisy image (Section~\ref{sec:lattice}). They are complementary constructions, not one lattice. The ciphertext equations currently use only a finite index set, a ring and a module. Independent operators $A_t$ in place of one $A$ would still act on a product over a discrete base; they would not by themselves make the bundle topologically nontrivial or curved. What would change the geometry is transport.

\subsection{Transport, holonomy and path dependence}
On a specified graph with vertex set $P$, assign an $\Rq$-linear isomorphism $U_e$ to each oriented edge $e$, with the reverse edge carrying $U_e^{-1}$. This is a discrete connection on $\EP$. Products along paths define parallel transport and products around closed paths define holonomy,
\[
\Hol(\gamma)=U_{e_\ell}\cdots U_{e_2}U_{e_1}\qquad\text{for a closed path }\gamma=e_1e_2\cdots e_\ell .
\]
Identity holonomy on all cycles means path-independent transport. Flatness can allow nontrivial monodromy; graph curvature needs additional face data (Appendix~\ref{app:curved}).

For an algebraic experiment, start from $b^\star=As^\star+e^\star$ and let $U_i$ be the transport applied after block $i$. Define
\[
V_0=I_k,\qquad V_{i+1}=U_iV_i,\qquad s^{(i)}=V_is^\star,\quad b^{(i)}=V_ib^\star,\quad e^{(i)}=V_ie^\star .
\]
If every $U_i\in GL_k(\Rq)$ commutes with $A$, then
\begin{equation}\label{eq:transported}
b^{(i)}=As^{(i)}+e^{(i)},
\end{equation}
so the cancellation identity of Theorem~\ref{thm:cancel} holds verbatim. Alternatively set $A^{(i)}=V_iAV_i^{-1}$, which gives $b^{(i)}=A^{(i)}s^{(i)}+e^{(i)}$ without commutation; operators publicly computable from the known state could be derived rather than stored, at additional computational cost. An $\Rq$-linear module map is not the same as a ring automorphism. Commutation is sufficient for the stated update; no classification of the centraliser is asserted, and scalar units $uI_k$ provide one simple commuting family.

Transport changes what exposure means. If $V_i$ is known and invertible, exposing $s^{(i)}$ exposes $s^\star=V_i^{-1}s^{(i)}$, and $V_i^{-1}b^{(i)}=b^\star$. Transported keys are correlated transforms of one sample, so they inherit neither the independent-family hybrid of Lemma~\ref{lem:hybrid} nor the independent-fibre exposure model; a successor would need a related-key and composition analysis.

\subsection{An integrality obstruction}
Correctness constrains transport more sharply than algebra does. Lemma~\ref{lem:fixedkey} charges a key through its squared norm, so a useful transport should act on integer coefficient representatives and should not inflate their length. Exact norm preservation leaves almost nothing.

\begin{proposition}[Integral isometries are signed permutations]\label{prop:isometry}
For every $N\ge1$, $O(N)\cap GL(N,\ZZ)$ is the group of signed permutation matrices, of order $2^NN!$.
\end{proposition}
\begin{proof}
Let $U$ be orthogonal with integer entries. Each column is an integer vector of Euclidean length one, hence of the form $\pm e_j$. Distinct columns are orthogonal, so they involve distinct $e_j$. Thus $U$ is a signed permutation matrix, and conversely every signed permutation matrix is orthogonal and unimodular.
\end{proof}

\begin{corollary}\label{cor:monomial}
A linear transport acting on integer coefficient representatives of $(s,e)\in\ZZ^{2kn}$ that preserves the squared norm $S$ of every key exactly is a signed permutation of coordinates. Among $\Rq$-linear maps commuting with $A$, the scalar monomials $\pm X^jI_k$ are of this kind, and they form a cyclic group of order $2n=512$.
\end{corollary}
\begin{proof}
The support of the key distribution contains every $e_i$ and every $e_i+e_j$, so preserving the norm of every key gives, by polarisation, an integer orthogonal matrix; the first statement is then Proposition~\ref{prop:isometry} with $N=2kn$. Multiplication by $X^j$ shifts coefficients negacyclically, which is a signed permutation, and $X^{2n}=1$ with $X^n=-1$.
\end{proof}

Every other integral transport changes some norms and therefore consumes noise budget. If $\|V\|_{\mathrm{op}}$ denotes the operator norm of a transport on integer representatives, then $S\bigl(V(\bar s,\bar e)\bigr)\le\|V\|_{\mathrm{op}}^2\,S(\bar s,\bar e)$, and the budget of the next subsection becomes a constraint on the holonomy group itself.

\subsection{A quantitative conditional noise budget}
Lemma~\ref{lem:fixedkey} needs no independence among the coefficients of the fixed key. It therefore applies to transported representatives $s',e'$ with $b=As'+e'\pmod q$, as long as block randomness is fresh and independent. If along a prescribed length-$L$ trajectory
\[
S_i=\|s^{(i)}\|_2^2+\|e^{(i)}\|_2^2\le S_{\max},
\]
then the conditional stream-failure bound is
\begin{equation}\label{eq:budget}
2Ln\exp\Bigl(-\frac{832^2}{\eta(S_{\max}+1)}\Bigr).
\end{equation}
For $L=64$, $n=256$, $\eta=2$ one has $2Ln=2^{15}$, so this is at most $2^{-128}$ if
\begin{equation}\label{eq:budgetcond}
\zbox{S_{\max}+1\le\frac{832^2}{2\cdot143\ln2}=3491.8466\ldots}
\end{equation}
In particular, $S_{\max}+1\le3491$ suffices. For $T=16$, the event that some untransported key of the family has $S\ge2400$ has probability below $2^{-194}$. A sufficient uniform condition on integer lifts is $\|V_i\|_{\mathrm{op}}^2\le3490/2400\approx1.454$, or $\|V_i\|_{\mathrm{op}}\le1.20589$. It is not necessary: realised key norms, rather than worst-direction expansion, enter \eqref{eq:budget}.

This is a sufficient certificate. Integral isometric holonomy lies in a signed-permutation subgroup; the scalar monomials need not exhaust all compatible isometries.

This is a correctness budget for realised representatives. It does not establish their norm bound, an unconditional failure probability or security of their distribution. If the norm condition can fail, its probability must be added, and enforcing it by rejecting keys changes the sampling law. Variance inflation alone does not justify reusing an earlier Bernstein bound when coefficients and dependencies change. Any useful extension must offer a functional or security benefit beyond relabelling and publicly invertible transformations.

\subsection{Curved Kähler geometry as a longer-term direction}\label{sec:curved}
The discrete picture has a continuous counterpart. Replace the finite base by a compact Kähler manifold $(M,\omega)$ and choose a \emph{holomorphic} vector bundle $E\to M$ with Hermitian metric $h$. Its Chern connection $\nabla$ is compatible with both structures. Keys may be values of a specified class of sections, and a walk becomes a path $c:[0,1]\to M$. Parallel transport $P_c:E_{c(0)}\to E_{c(1)}$ is the continuous analogue of the products $V_i$. Infinitesimal loop holonomy is controlled by the curvature $F_\nabla$; the Ambrose--Singer theorem~\cite{ambrose} identifies the holonomy Lie algebra through transported curvature values. The real first Chern class $c_1(E)_\RR=[\tfrac{i}{2\pi}\operatorname{tr}F_\nabla]$, when nonzero, obstructs topological triviality. Its vanishing does not imply a trivial bundle, zero curvature or trivial holonomy. Appendix~\ref{app:curved} separates these possibilities and gives concrete testbeds.

Three curvature regimes illustrate what such a continuous theory would have to meet.

\paragraph{1. Positive curvature: a continuous analogy.} CLWE uses a secret direction on $S^{n-1}$ and noisy samples related to parallel hyperplanes~\cite{clwe}. Reductions relate specified LWE and CLWE distributions, including classical hardness results~\cite{clwe-hard}. These depend on parameters and do not attribute hardness to curvature. The real sphere is not generally Kähler; this is an analogy, not a Kähler-bundle instantiation or a concrete security comparison.

\paragraph{2. Negative curvature: hyperbolic surfaces.} On a compact curve of genus $g\ge2$, unitary representations give flat bundles. Narasimhan--Seshadri identifies the associated polystable holomorphic bundles of degree zero~\cite{ns}. These can be holomorphically nontrivial while smoothly trivial. Nonabelian monodromy gives path dependence without curvature. Norm-preserving integral coefficient matrices are signed permutations (Proposition~\ref{prop:isometry}); other transports need a realised-noise analysis. Appendix~\ref{app:curved} develops these distinctions.

\paragraph{3. Ricci-flat: Calabi--Yau.} Yau's theorem supplies Ricci-flat Kähler metrics when the real first Chern class vanishes~\cite{yau}. Bochner makes harmonic one-forms and holomorphic tangent fields parallel; full holonomy $SU(d)$, $d\ge2$, forces them to vanish. This restricts \emph{holomorphic or parallel} sections, not smooth sections or finite key assignments. Smooth sections can interpolate prescribed values at finitely many points. Auxiliary bundles are a design choice, not a mathematical necessity. The earlier cryptographic failure was the scalar exposure in Section~\ref{sec:earlier}.

\medskip
Metric curvature and monodromy are distinct. A local system has flat transport even when the underlying manifold is curved: a torus fibration over a smooth locus $Y^\circ$ may supply a cohomology local system with monodromy $\rho:\pi_1(Y^\circ)\to GL(d_{\mathrm{loc}},\ZZ)$, whose reduction modulo $q$ yields finite modules and candidate transport maps. The Strominger--Yau--Zaslow picture motivates such fibrations~\cite{syz} without asserting that fibre metrics are flat, and Manin's treatment of mirror symmetry and quantisation of abelian varieties provides the closest arithmetic setting~\cite{manin-mirror}. None of these references substitutes for a cryptographic assumption.

A concrete research path follows: (i)~choose a graph or cell complex on $P$ and an arithmetically compatible action; (ii)~control realised norms and establish a correctness bound, using \eqref{eq:budgetcond} when applicable; (iii)~distinguish holonomy, curvature and monodromy; (iv)~give a related-key analysis replacing Lemma~\ref{lem:hybrid}; (v)~demonstrate a functional or security benefit. Failing the sufficient test \eqref{eq:budgetcond} does not rule out a construction; another analysis may suffice. Appendix~\ref{app:curved} develops nontrivial-bundle examples, compatible local equations and further integrality constraints.

\Needspace{7\baselineskip}
\section{Reproducibility and independent verification}\label{sec:repro}
This release provides a self-contained LaTeX manuscript and a separate numerical-verification supplement. The full encryption software is not distributed with this manuscript. Filenames associated with the reference implementation below identify retained development artefacts, not files supplied by the LaTeX source. The independent verifier recomputes the displayed numerical bounds and hash vectors. The end-to-end implementation records in Section~\ref{sec:fixture} were not rerun in this revision, and are not independently reproducible from the manuscript alone. No additional benchmark or cryptanalytic experiment is claimed.

\subsection{Three levels of evidence}
\paragraph{Specification.} Sections~\ref{sec:keygen}--\ref{sec:encdec} specify the ring algorithms; Appendix~\ref{app:profile} fixes the hash and wire profile. Its hash vectors are retained unchanged. Section~\ref{sec:fixture} records a separate implementation fixture with its own generator and draw order; a digest from a different generator need not agree. This revision independently verifies the hash profile, not the recorded end-to-end implementation digest.

\paragraph{\color{zred}Numerical verification.} \texttt{anc/verify\_math.py} imports no code from the prototype. It constructs the exact integer coefficients of \eqref{eq:Slaw} and checks their total, mean and variance; exhaustively checks all bit/error combinations with $|\delta_c|\le831$; verifies the boundary counterexample; computes membership-region cardinality and binomial tails; reproduces the hash vectors and both walks; checks exact rational prefix laws; and exercises the integer cancellation identity on an independent ring-arithmetic fixture. Exact arithmetic is used for finite combinatorial laws; logarithms and exponentials are numerical evaluations of proved expressions.

Writing $C(z)=(6+8z+2z^4)^h=\sum_sc_sz^s$, the identity $(6+8z+2z^4)C'=h(8+8z^3)C$ gives the efficient exact recurrence
\begin{equation}\label{eq:recurrence}
c_0=6^h,\qquad6s\,c_s=8(h-s+1)c_{s-1}+2(4h-s+4)c_{s-4},
\end{equation}
with coefficients outside the range set to zero. Division is exact, the tail probability is $16^{-h}\sum_{s\ge S_0}c_s$, and the verifier checks the normalisation $\sum_sc_s=16^h$ before using this law.

\paragraph{\color{zred}Reference implementation.} \texttt{anc/zsigil.py} supplies key generation, state functions, ring operations, encryption, decryption, 12-bit packing, strict ciphertext parsing, candidate-table decryption and membership counts. Operational sampling uses Python's \texttt{secrets} interface to the operating-system random source; \texttt{ShakeRNG} is a public deterministic fixture generator used only for reproducibility. \texttt{test\_reference.py} checks boundary message lengths, leading padding, table decryption, canonical residues and frames, malformed-ciphertext rejection, a membership fixture, bypass of a known plaintext block, hash vectors and exhaustive enumeration of a two-stream toy model. Tests establish behaviour on those cases, not negligible failure probabilities or security. The implementation uses quadratic convolution and is not constant-time.

\subsection{Recomputed figures and sizes}
All constants below were recomputed for this version from the formulas as printed. The variance is 1537. The 64-block averaged Bernstein bound is $5.18075\ldots\times10^{-53}$. At $S_0=2400$ the base-2 logarithms of the two terms of \eqref{eq:streambound} are $-194.346565\ldots$ for the exceptional-family term and $-193.055861\ldots$ for fresh noise; their sum has logarithm $-192.561444\ldots$. The Gaussian-heuristic scale is 547.161 against $\sqrt{1536}=39.192$. Exhaustive enumeration of the 1024 outcomes of two independent binary streams, each with two leading blocks and three framed blocks, gives mean maximum framed recovery $427/1024$, matching \eqref{eq:padmax}. At $p=1/2$, $L=64$, $W=1024$ the exact unpadded expected maximum prefix is $10.33345\ldots$ blocks, and with $\nu=8$, $\ell=64$, $W=256$ the expected maximum framed recovery is $0.85504\ldots$.

At $L=64$ the ciphertext has $72+64\cdot1536=98376$ bytes, and the full $T=16$ decoded table has 32768 bytes. A 2040-byte user payload produces exactly 64 framed blocks when $\nu=0$. Expansion is 48 per framed block, before header and length overhead.

\begin{table}[htbp]
\centering\small
\caption{Rounded theoretical values in the independent-selector direct-prefix model, $T=16$, $L=64$. These are finite-sum predictions, not measured recovery rates.}
\label{tab:prefix}
\begin{tabular}{@{}rrr@{}}
\toprule
Exposed keys out of 16 & Expected prefix & $\Pr(N=64)$\\
\midrule
1 & 0.0667 & $8.64\times10^{-78}$\\
4 & 0.3333 & $2.94\times10^{-39}$\\
8 & $1-2^{-64}$ & $5.42\times10^{-20}$\\
12 & 3.0000 & $1.01\times10^{-8}$\\
14 & 6.9986 & $1.94\times10^{-4}$\\
15 & 14.7589 & $1.61\times10^{-2}$\\
16 & 64 & 1\\
\bottomrule
\end{tabular}
\end{table}

\subsection{Deterministic end-to-end fixture}\label{sec:fixture}
The retained implementation record describes a demonstration with $T=2$, $\nu=0$, a 32-byte all-zero fixture seed and the ASCII message
\begin{center}\texttt{Z-Sigil/spec-v2 deterministic known-answer message.}\end{center}
The record reports two blocks and 3144 ciphertext bytes, with SHA-256 digest
\begin{center}\texttt{3f61bb8f44c02f2fc82ea4e8c194f89d2}\\ \texttt{b7a153c338831811b4557ea96f84fa6}\end{center}
where the line break is not part of the value. Public-key and secret-fixture digests and the table-decryption result are in \texttt{reproduction.json}. All fixture secrets are public and must not be used for operational encryption.

The fixture generator concatenates 64-byte SHAKE256 outputs on
\[
\texttt{ZSIGIL-FIXTURE}\,\|\,\mathrm{seed}\,\|\,\mathrm{U64}(\mathrm{counter}),
\]
incrementing the counter per chunk and discarding unused suffix bytes on each request. Key generation draws $A$ in row order, then $s_t,e_t$ in increasing $t$ and component/coefficient order. Framing draws leading blocks, encryption draws the nonce, and each block draws $r,f,g$ in that order. The source specifies individual sampling requests, which are needed to reproduce this implementation fixture exactly.

\subsection{Running the package and interpreting success}
For the numerical-verification supplement supplied with this revision:
\par\medskip\noindent\begin{minipage}{\linewidth}
\begin{lstlisting}
python3 anc/verify_math.py
\end{lstlisting}
\end{minipage}\par\medskip
This writes \texttt{anc/verification\_review.json} and needs only Python~3.11 or later with its standard library. The historical command \texttt{python3 anc/run\_all.py} instead refers to the separate, undistributed development package; it cannot be run from this LaTeX submission. Compile the manuscript with a TeX installation:
\par\medskip\noindent\begin{minipage}{\linewidth}
\begin{lstlisting}
latexmk -pdf -interaction=nonstopmode Z_Sigil_v2.tex
\end{lstlisting}
\end{minipage}\par\medskip The figures are drawn by TikZ/pgfplots from data embedded in the LaTeX source, so compilation needs no external figure files. Successful execution establishes reproducibility, not cryptographic validation; the checksum manifest helps detect accidental changes and is not a signature.

A separate retained default-family fixture, \texttt{extended\_check.py}, is recorded with $T=16$, a 2040-byte payload and 64 ciphertext blocks. The record reports that ordinary decryption and all 1024 candidate decodings followed by traversal recover the same message; the ciphertext has 98376 bytes and the decoded table 32768 bytes. This is a reproducibility check, not a failure-rate estimate.

If compression introduced reconstruction errors $\Delta u,\Delta v$, the residual would change by $\Delta v-\ip{s_t}{\Delta u}$, with the conservative representative bound
\[
\|\Delta v-\ip{s_t}{\Delta u}\|_\infty\le\|\Delta v\|_\infty+\|s_t\|_1\|\Delta u\|_\infty .
\]
An observed finite-run maximum is not a justified tail bound or compression budget.

\Needspace{7\baselineskip}
\section{Relation to established work}
The arithmetic belongs to the LWE, Ring-LWE and Module-LWE line~\cite{regev,lpr,ls}, using the public-key block pattern of~\cite{lp,kyber-spec,kyber-eurosp} and small-coefficient sampling familiar from constructions such as NewHope~\cite{newhope}. It introduces no new worst-case lattice reduction, and concrete attack estimation for these distributions and the shared matrix remains necessary~\cite{aps}.

Online encryption~\cite{online}, key privacy~\cite{keypriv}, KEM combiners~\cite{combiners}, threshold sharing~\cite{shamir} and all-or-nothing transforms~\cite{rivest,boyko} supply different comparisons. Forward-secure public-key encryption~\cite{chk} and ratcheting protocols~\cite{acd} evolve key material over time or over a conversation; Z-Sigil's family is fixed and the plaintext moves an access pointer. It has no claimed forward secrecy or post-compromise recovery, and the chain denotes a recurrence, not an established ratchet security property.

The geometric references~\cite{manin-theta,manin-mirror,syz,ambrose,ns,yau} and the CLWE line~\cite{clwe,clwe-hard} motivate the representation and the future transport programme. The current scheme is specified by finite algebra and its byte profile, and its confidentiality proof does not depend on the metric or curvature. The earlier proposal~\cite{rondelli1} is superseded for the reasons of Section~\ref{sec:earlier}.

\Needspace{7\baselineskip}
\section{Validation roadmap}\label{sec:roadmap}
\begin{enumerate}
\item Estimate classical and quantum attacks on the exact small-secret distributions and the shared-matrix family, including the $(k+T)$-row instance behind $\varepsilon_3$. Parameter resemblance to ML-KEM is not a substitute.
\item Independently review the composition proof and its losses. Define an unrestricted partial-exposure game including membership tests, known plaintext, guessing and leakage; do not substitute the direct-prefix law for its security bound.
\item Study parallel time and work with candidate tables, preprocessing, speculative branches and quantum queries where appropriate. Any lower bound must name its model and resource limits.
\item Compare a shared $A$ with independent $A_t$. Explicit storage increases by $(T-1)k^2$ ring elements, but this alone does not prove a $T$-fold attack cost.
\item Investigate compatible transport following the five steps of Section~\ref{sec:curved}, with a controlled key/noise distribution, a related-key analysis and a demonstrable benefit. Distinguish curvature, monodromy and finite actions.
\item Develop cross-language byte tests and independent implementations. Authentication or IND-CCA requires a separate protocol and proof; compression and implementation hardening require separate correctness and side-channel analyses.
\end{enumerate}

\Needspace{7\baselineskip}
\section{An invitation to verify}
The mathematical claims and executable checks invite different kinds of independent review. A numerical verifier can recompute displayed constants, and Appendix~\ref{app:profile} lets a separate implementation reproduce the hash vectors byte for byte. It cannot certify a security reduction. The author welcomes independent recomputation and disagreement where results differ. Section~\ref{sec:repro} distinguishes the checks repeated for this revision from the retained implementation records.

Three research directions follow. First, concrete cryptanalysis of the shared-matrix family and its exact small-secret distributions is needed before any parameter recommendation. Second, a Fujisaki--Okamoto-style successor needs its own construction and proof; authentication, implementation hardening and protocol analysis are further requirements for deployment. Third, the transport programme in Section~\ref{sec:transport} and Appendix~\ref{app:curved} supplies explicit arithmetic compatibility conditions, integral-transport obstructions and sufficient noise certificates. A construction that misses a sufficient certificate may still admit a sharper analysis. Appendix~\ref{app:seriality} similarly separates the exact restricted fragment-recovery gain from the still-open question of additional adversarial cost.

Corrections, failed reproductions, attacks and counterexamples are all welcome before the architecture is built upon. Correspondence may be addressed to
\begin{center}\href{mailto:andrearondelli85@gmail.com}{\texttt{andrearondelli85@gmail.com}}\quad or through\quad\href{https://www.zsigil.com}{\texttt{www.zsigil.com}}.\end{center}

\Needspace{7\baselineskip}
\section{Conclusion}
Z-Sigil combines a fixed family of module-lattice keys with a hash state updated from successive plaintext blocks. The state selects the key for each block; after recovering that block, the receiver updates the state to determine the next selection. A fibre bundle over a finite set of points on a flat K\"ahler torus describes the key family and the resulting walk.

The paper specifies key generation, encryption and decryption, proves correct recovery under an explicit noise condition, and bounds the probability of decoding failure. It also provides an IND-CPA confidentiality reduction for the complete chain under stated decisional Module-LWE assumptions. This reduction applies to quantum adversaries when those assumptions hold against quantum algorithms, with classical public keys, messages and ciphertexts. The construction therefore has a conditional post-quantum confidentiality result, although no concrete security level is established for the proposed parameters.

The additional protection that chaining or geometric structure may provide remains an open research question. The restricted key-exposure results do not establish security against unrestricted attacks, and the construction provides neither authentication nor chosen-ciphertext security. The pseudocode, byte-level specification and test vectors provide a basis for independent implementation and verification. The next steps are concrete cryptanalysis, comparative evaluation of chaining, and extensions in which geometric transport has an explicit computational role. This work provides a specified construction and testable questions for those investigations.

\paragraph{\color{zred}Acknowledgements.} The author thanks the interlocutors who contributed critical feedback during the cryptanalysis and redesign. Responsibility for the construction and any remaining errors is the author's.

\clearpage
\appendix
\Needspace{7\baselineskip}
\section{Concrete hash and encoding profile}\label{app:profile}
\spec{} fixes $n=256$, $k=3$, $q=3329$, $\eta=2$, permits $1\le T\le256$ and $0\le\nu\le256$, and uses the following conventions. The default is $T=16$, $\nu=0$.

\subsection{Tuple encoding and domain separation}
Let U32, U64 be unsigned big-endian four- and eight-byte encodings. For byte strings $a_0,\dots,a_j$ of lengths below $2^{32}$, set
\[
F(a_0,\dots,a_j)=\mathop{\big\Vert}_{h=0}^{j}\bigl(\mathrm{U32}(|a_h|)\,\|\,a_h\bigr).
\]
The fixed 24-byte context is
\[
c=\mathrm{U32}(n)\|\mathrm{U32}(k)\|\mathrm{U32}(q)\|\mathrm{U32}(\eta)\|\mathrm{U32}(T)\|\mathrm{U32}(\nu).
\]
Let $D(lab;a_1,\dots,a_j;z)$ be the first $z$ output bytes of SHAKE256~\cite{fips202} on
\[
\texttt{Z-Sigil/spec-v2}\,\|\,\texttt{00}\,\|\,F(c,lab,a_1,\dots,a_j).
\]
The prefix and labels are literal ASCII; \texttt{00} denotes one zero byte, not two ASCII zero characters. Display spacing and line breaks add no bytes, and $L$ is not a context field.

\subsection{States, masks and index selection}
States, nonces and blocks contain 32 bytes, with bits least-significant first within each byte. Define
\begin{align*}
H(\lab{INIT},iv)&=D(\lab{INIT};iv;32),\\
H(\lab{CHAIN},i,x,m)&=D(\lab{CHAIN};\mathrm{U64}(i),x,m;32),\\
H(\lab{MASK},x)&=D(\lab{MASK};x;32).
\end{align*}
For \lab{FIBRE}, return 0 if $T=1$. Otherwise put $B=T\lfloor2^{256}/T\rfloor$. For $j=0,\dots,2^{64}-1$, interpret $D(\lab{FIBRE};x,\mathrm{U64}(j);32)$ as a big-endian integer $y$ and return $y\bmod T$ at the first $y<B$; signal an error on exhaustion.

For supported power-of-two $T$, including the default, $B=2^{256}$ and the first trial always succeeds, so Theorem~\ref{thm:composition} applies directly to these bounded, total selectors. For other $T$, aborts must be modelled separately; efficient totality is not obtained merely from a finite but enormous counter loop. Rejection removes modulo bias for independent uniform trial digests, conditional on success, but a deterministic hash does not prove independence of real selectors.

\subsection{Framing, packing and strict parsing}
Prepend $\nu$ independent uniform blocks, then frame $\mathrm{U64}(|msg|)\|msg$, appending the fewest zero bytes needed for a multiple of 32. Require $L<2^{64}$; the prototype enforces $\nu+1\le L\le32768$. Reject excessive declared message lengths, nonzero padding and extra whole zero blocks; the last check enforces minimal padding.

Pack coefficients $a_j\in\{0,\dots,3328\}$ into the integer $\sum_{j=0}^{255}a_j2^{12j}$, serialised little-endian in exactly 384 bytes. Vectors list polynomials in component order and matrices in row order. Raw public-key payloads list $A$ followed by $b_0,\dots,b_{T-1}$; private-key payloads list $s_0,\dots,s_{T-1}$. Consistent parameter metadata must accompany use of these raw payloads.

A ciphertext consists of the eight-byte ASCII tag \texttt{ZSIGIL02}, the 24-byte context, $\mathrm{U64}(L)$, the 32-byte nonce and $L$ packed pairs, each listing the $k$ polynomials of $u_i$ and then $v_i$; the header is 72 bytes. Reject tag or context mismatches, noncanonical coefficients, length mismatches and trailing bytes, and bound allocations from untrusted lengths. These checks provide no cryptographic integrity.

Uniform coefficients of $A$ use independent 12-bit draws with rejection of values $\ge3329$. A $\CBD_2$ coefficient is $a_1+a_2-a_3-a_4$ for four fresh unbiased bits. All coefficients of $s_t,e_t,r,f,g$ require independent draws. A deterministic experiment must also specify its generator and draw order.

\subsection{Known-answer vectors}\label{app:kat}
At $T=16$, $\nu=0$, with a zero nonce and a zero first block, $t_0=14$ and
\begin{center}\small
\begin{tabular}{@{}ll@{}}
\toprule
Value & Hexadecimal bytes (concatenate both lines)\\
\midrule
$x_{-1}$ & \texttt{48ce8f836fd66e1295e1be7979be18ca}\\ & \texttt{e7364b1d0305cd6a490cd88a45c2289f}\\[2pt]
$\kappa_0$ & \texttt{b3615e308fa16c9d42f1bf9f2c9feb7f}\\ & \texttt{b3a910b0bd39776d31b5ce2a681bf4d0}\\[2pt]
$x_0$ & \texttt{1bae58453673bf237f54441023dff875}\\ & \texttt{1f074fc76dcb54078b78b50042cd8369}\\
\bottomrule
\end{tabular}
\end{center}
The walks of Figure~\ref{fig:trace} are
\[
(14,2,3,13,3,5,1,14)\quad\text{and}\quad(14,2,3,13,7,7,13,1).
\]
These require no lattice operations. They test domain separation, tuple encoding, bit order and the one-block delay, and can be reproduced independently from this appendix; all values above were recomputed from the specification for this version.

\Needspace{40\baselineskip}
\section{Claim-to-artefact checklist}\label{app:checklist}
\begin{center}\small
\captionof{table}{What each artefact establishes. Successful execution cannot replace review of a cryptographic reduction.}
\begin{tabularx}{\textwidth}{@{}>{\raggedright\arraybackslash}p{3.6cm}>{\raggedright\arraybackslash}p{2.6cm}>{\raggedright\arraybackslash}X@{}}
\toprule
Claim or object & Evidence & How to check\\
\midrule
Block and chained inversion & proof & Propositions~\ref{prop:blockinv} and~\ref{prop:chaininv}; round trips\\
Cancellation and margin & proof and arithmetic & Theorem~\ref{thm:cancel}; exhaustive margin and integer fixtures\\
Family-conditioned failure & conditional upper bound & Theorem~\ref{thm:stream}; exact norm coefficients and evaluated exponentials\\
Augmented-lattice membership & proof & Proposition~\ref{prop:planted}; substitute the key relation into \eqref{eq:auglattice}\\
547 versus 39 scales & heuristic & evaluate the Gaussian heuristic; not an attack estimate\\
Chained IND-CPA & conditional reduction & Theorem~\ref{thm:composition}; human proof review required\\
Prefix and padding laws & exact model formulas & rational checks; two-stream exhaustive enumeration\\
Membership constants & fixed-index experiment & \eqref{eq:fn}--\eqref{eq:fp}; verifier and prototype fixture\\
No-witness probability & ideal calculation & \eqref{eq:q0}; excludes accumulated test errors\\
Integral isometries & proof & Proposition~\ref{prop:isometry}\\
Noise budget & conditional bound & \eqref{eq:budgetcond}; arithmetic only\\
State walk & hash fixture & Appendix~\ref{app:kat}; independent verifier and prototype\\
Wire format and round trips & specification; retained test records & Appendix~\ref{app:profile}; development records not rerun or distributed here\\
Numerical figures & derived plots & data embedded in the source; \texttt{verification\_review.json}\\
Concrete security, authentication, unrestricted exposure, sequential lower bounds & not established & separate analyses; Section~\ref{sec:roadmap}\\
\bottomrule
\end{tabularx}
\end{center}

\Needspace{7\baselineskip}
\section{What seriality gains, and how to measure it}\label{app:seriality}
The hash recurrence is an architectural feature to evaluate against specified baselines. This appendix separates three questions: how many fragments an exposed-key procedure recovers, how much work a complete-key decryptor performs, and whether unrestricted attackers incur additional cost. The first admits an exact comparison. The second admits explicit upper-bound algorithms. The third remains open; none of the following ratios is an additional number of security bits.

\subsection{An exact comparison with independent key-labelled blocks}
Use the same family size $T$, exposed set $K$, block count $L$ and exposed fraction $p=|K|/T$ as Definition~\ref{def:prefix}. Couple both procedures to the same independent uniform indices $t_0,\dots,t_{L-1}$. Ignore decoding failures. The baseline publishes the index of each block and has no unavailable prefix-dependent mask; it recovers every block whose index lies in $K$. Let $B$ count these recovered blocks. Z-Sigil's stipulated direct-prefix procedure recovers $N$ consecutive initial blocks, then stops.

\begin{proposition}[Fragment suppression, not a full-recovery advantage]\label{prop:fragment}
In this coupled model, $N\le B$ pointwise, $B\sim\mathrm{Binomial}(L,p)$, and
\begin{gather}
\E B=Lp,\qquad\E N=\sum_{j=1}^Lp^j=\frac{p(1-p^L)}{1-p},\quad0\le p<1,\label{eq:fragmean}\\
\Pr(B=L)=\Pr(N=L)=p^L.\label{eq:fragfull}
\end{gather}
For $0<p<1$, the ratio of expected recovered-block counts is
\begin{equation}\label{eq:gfrag}
G_{\mathrm{frag}}:=\frac{\E B}{\E N}=\frac{L(1-p)}{1-p^L}.
\end{equation}
\end{proposition}
\begin{proof}
Put $I_i=\mathbf 1_{\{t_i\in K\}}$. Then $B=\sum_iI_i$ and $N$ is the length of the initial run of ones. Hence $N\le B$, and both counts equal $L$ exactly when every $I_i=1$. Independence gives the binomial law and the probability $p^L$. The tail-sum identity gives $\E N$ and then the ratio. Repeated use of a key does not invalidate this argument: $K$ is fixed and the selectors, rather than the key values, are independent.
\end{proof}

At $T=16$, $|K|=8$ and $L=64$, the baseline recovers 32 blocks on average, while the prefix procedure recovers $1-2^{-64}$; thus $G_{\mathrm{frag}}=32/(1-2^{-64})$. Both recover the complete stream with probability $2^{-64}$. The demonstrable gain is less recovered fragment volume under the stated stopping rule, not a better full-stream recovery probability in this comparison.

\subsection{Known plaintext and the cost of guessing}
The comparison in \eqref{eq:fragmean} changes when missing blocks are predictable. Knowing the current state and a candidate value of $m_i$ lets a party compute $x_i$ directly, without the missing secret. Membership tests and candidate-message attacks from Section~\ref{sec:exposure} can then test positions beyond the first unavailable key. A framing block, file header, repeated content or external copy of a message can supply such candidates. These are part of the threat model, not implementation faults.

If the unknown block has conditional min-entropy $h$ given the adversary's view, the success probability of the best single guess is at most $2^{-h}$ by definition. This is not a lower bound of $2^h$ on attack work: multiple guesses, structured distributions, verifiable side information and quantum queries require their own analyses. An $n$-bit block need not contain $n$ bits of conditional entropy. Leading random blocks therefore need an explicit conditional-entropy argument before they can support an unrestricted exposure claim.

The finite family also matters. Under independent selectors, the number $D_L$ of distinct keys visited satisfies
\begin{equation}\label{eq:distinct}
\E D_L=T\bigl(1-(1-1/T)^L\bigr),\qquad D_L\le T.
\end{equation}
Indeed, each fixed index is visited with probability $1-(1-1/T)^L$. At $T=16$, $L=64$, the expectation is about 15.74. Sixty-four blocks do not create sixty-four independent long-term secret keys. Offline analysis of $(A,b)$ can target the entire family independently of the plaintext walk.

\subsection{Work and depth for a party holding all secrets}
Let $W_{\mathrm{dec}},D_{\mathrm{dec}}$ be the work and depth of one lattice decoding. Let $W_H,D_H$ aggregate selector evaluation, mask generation, XOR, \lab{CHAIN} evaluation and lookup for one traversal step. These symbols count specified computations; they are not measured timings or lower bounds for arbitrary algorithms. For the default power-of-two $T$, selector evaluation has bounded cost as in Appendix~\ref{app:profile}.
\begin{center}\small
\begin{tabular}{@{}llll@{}}
\toprule
Procedure & Work & Parallel depth & Decoded-table storage\\
\midrule
Independent labelled baseline & $LW_{\mathrm{dec}}$ & $D_{\mathrm{dec}}$ & not required\\
Prescribed chain & $L(W_{\mathrm{dec}}+W_H)$ & $L(D_{\mathrm{dec}}+D_H)$ & not required\\
Candidate table, then walk & $LTW_{\mathrm{dec}}+LW_H$ & $D_{\mathrm{dec}}+LD_H$ & $LTn$ bits\\
\bottomrule
\end{tabular}
\end{center}
Depths assume enough processors. The table omits input/output storage and scheduling overhead. At the default example, the table method uses 1024 rather than 64 lattice decodings and 32768 bytes of decoded entries. This is the cost of that explicit strategy, not a proof that every attacker must pay a factor~16. The legitimate receiver also pays the chain dependency; an attacker may amortise preprocessing over many messages.

\subsection{An experiment that could establish a practical benefit}
A useful evaluation must hold the underlying arithmetic and resources fixed. Compare: (i)~independent labelled blocks under the same family; (ii)~a single-key or constant-selector variant; (iii)~the full plaintext-fed selector and mask; and (iv)~a conventional KEM plus authenticated symmetric encryption for operational cost. The first three isolate architectural effects; the fourth supplies the practical baseline of Section~\ref{sec:baselines}.
\begin{enumerate}
\item \emph{Declare the task.} Measure fragment recovery, full-message recovery, distinguishing advantage or elapsed decryption time separately. A change in one is not automatically a change in another.
\item \emph{Declare disclosure and side information.} Fix $K$ before the experiment, then separately study adaptive exposure; include known prefixes, low-entropy candidate sets, fresh random payloads and hidden leading random blocks. Keep training messages distinct from test messages.
\item \emph{Run competing procedures.} Include the direct decryptor, a known-plaintext bypass, membership tests, candidate-message traversal and complete candidate tables. Count failed tests, guesses and all precomputation.
\item \emph{Control resources.} Fix processor count, memory, total work, public-key family and message distribution. Report preprocessing time separately and amortised over a specified number of messages. Count lattice products and hash calls as well as wall-clock time.
\item \emph{Report uncertainty.} Repeat independent trials and report distributions and confidence intervals for empirical results. Compare simulations of the restricted model with its exact formulas. A measured slowdown of one implementation is not an adversarial lower bound.
\item \emph{State a success criterion in advance.} A practical benefit requires improved attack success versus time/memory at an acceptable legitimate-user cost, after the stronger procedures are included. If it disappears under candidate testing or preprocessing, report that outcome.
\end{enumerate}
No such end-to-end comparative benchmark is asserted in this release. The exact restricted-model comparison \eqref{eq:fragmean}--\eqref{eq:gfrag} and the candidate-table work counts are established; additional security against unrestricted classical or quantum adversaries is not. Theorem~\ref{thm:composition} remains a conditional confidentiality result: its proof works even with constant public state functions, so it cannot by itself certify a security improvement caused by the hash recurrence.

\paragraph{A sharper research question.} For which message distributions, conditional entropy assumptions, key-isolation mechanisms and time--memory budgets does plaintext-fed access reduce an adversary's recoverable information more than it increases legitimate processing cost? A sequential-work claim would additionally need a precise classical or quantum query model, including preprocessing, speculation and possible state collisions. This is the criterion by which the seriality proposal should be tested.

\Needspace{7\baselineskip}
\section{Development toward curved, nontrivial key bundles}\label{app:curved}
The current construction uses a product over a finite base. A curved manifold and a nontrivial bundle are meaningful future ingredients only if the arithmetic and security experiment retain information that cannot be removed by a public change of coordinates. This appendix makes the choices and obstacles explicit. It specifies research directions, not a new cryptosystem or an extension of Theorem~\ref{thm:composition}.

\subsection{Three kinds of nontriviality}
The base metric, the bundle and its connection are different objects. A curved base can carry a trivial bundle. A topologically trivial bundle can carry a nonflat connection. A flat connection can have nontrivial monodromy around noncontractible loops. On a graph, products around all cycles equal to the identity mean trivial holonomy and path-independent transport within each connected component. To define discrete curvature, one must additionally specify faces of a cell complex and impose or measure holonomy around their boundaries.
\begin{center}\small
\begin{tabularx}{\textwidth}{@{}>{\raggedright\arraybackslash}X>{\raggedright\arraybackslash}X>{\raggedright\arraybackslash}X@{}}
\toprule
Testbed & Geometric content & What must be solved\\
\midrule
$M=\mathbb{CP}^1$, $E=\mathcal O(1)\oplus\mathcal O^{k-1}$ & Compact Kähler base; $\int_Mc_1(E)=1$, so $E$ is topologically nontrivial & Arithmetic charts, coefficient lattice, compatible public operators and noise\\
Genus-$g\ge2$ compact curve; degree-zero unitary bundle & Nontrivial holomorphic structure and possibly nonabelian monodromy with a flat connection & Distinguish holomorphic nontriviality from smooth topology; analyse correlated keys\\
Ricci-flat Kähler base with an auxiliary bundle & Base Ricci flatness need not make the auxiliary bundle or its Chern connection flat & Specify the bundle, section class, integral model and cryptographic use of transport\\
\bottomrule
\end{tabularx}
\end{center}
For the first testbed, the degree-one transition of $\mathcal O(1)$ has winding number one on the overlap of the standard sphere charts, giving its first Chern number. This is a concrete example of nontriviality, not a hardness statement. Global holomorphic endomorphisms and sections are constrained; they cannot be replaced silently by unrestricted independent local matrices and secrets.

In the second testbed, the Narasimhan--Seshadri correspondence~\cite{ns} links unitary representations to polystable holomorphic bundles of degree zero. Over a compact surface, degree-zero complex vector bundles are topologically trivial as smooth complex bundles, though they can be holomorphically nontrivial. Thus this example supplies nontrivial monodromy without the topological nontriviality of the first example.

\subsection{Local arithmetic and the dual pairing}
Choose an explicit cover and transition functions before sampling keys. For a finite arithmetic model over $\Rq$, let $G_{21}\in GL_k(\Rq)$ identify local coordinates of chart~1 with those of chart~2 on an overlap, with $G_{31}=G_{32}G_{21}$. The local key relation is preserved by
\begin{equation}\label{eq:chart}
s_{(2)}=G_{21}s_{(1)},\quad e_{(2)}=G_{21}e_{(1)},\quad b_{(2)}=G_{21}b_{(1)},\quad A_{(2)}=G_{21}A_{(1)}G_{21}^{-1}.
\end{equation}
Here $1,2,3$ label charts, not the fibre indices $t$ of the present scheme. Vectors paired with keys must transform dually. With $G=G_{21}$, set
\begin{equation}\label{eq:dual}
u_{(2)}=G^{-\top}u_{(1)},\quad r_{(2)}=G^{-\top}r_{(1)},\quad f_{(2)}=G^{-\top}f_{(1)},\quad v_{(2)}=v_{(1)}.
\end{equation}
Then $u_{(2)}=A_{(2)}^\top r_{(2)}+f_{(2)}$ and
\[
b_{(2)}^\top r_{(2)}=b_{(1)}^\top r_{(1)},\qquad s_{(2)}^\top u_{(2)}=s_{(1)}^\top u_{(1)}.
\]
Thus both encryption and the cancellation pairing are coordinate-consistent. The transpose is the algebraic dual over the commutative ring $\Rq$; it is not a Hermitian adjoint. A continuous Hermitian interpretation does not change the finite scheme's bilinear convention.

These identities alone are insufficient. General $G$ does not preserve coefficientwise CBD sampling. A globally consistent noise distribution, its integer representatives and its reduction hypothesis must be specified. Sampling new independent CBD coefficients in every chart and simultaneously requiring \eqref{eq:dual} would generally be inconsistent. Sharing a public operator across charts imposes the additional commutation condition from Section~13.2.

There is also no automatic reduction of arbitrary complex transition functions modulo $q$. One needs a chosen integral or algebraic model, a reduction map at an admissible modulus or prime ideal, and control of denominators and bad fibres. Alternatively one may discretise selected paths numerically, but then approximation, quantisation and coordinate errors enter the decoding residual and require a fresh bound.

\paragraph{Restriction to finitely many points.} Restricting any vector bundle to a finite discrete set of sampling points gives a product after choosing a basis in each fibre. Topological nontriviality of the ambient bundle therefore does not survive as nontriviality of that restricted bundle by itself. Any algorithmic content must also retain specified transition, transport, section or consistency data. If those data can be removed by publicly invertible local changes without affecting the attack problem, the construction has only been re-expressed.

\subsection{Two constraints on exact transport}
\begin{proposition}[An integer norm gap]\label{prop:normgap}
Let $V$ be a square nonsingular integer matrix acting linearly on coefficient representatives. If $\|V\|_{\mathrm{op}}<\sqrt2$, then $V$ is a signed permutation matrix.
\end{proposition}
\begin{proof}
Every column $Ve_j$ is a nonzero integer vector of norm less than $\sqrt2$, hence has squared norm one and equals a signed coordinate vector. Nonsingularity forces different columns to use different coordinates.
\end{proof}

Consequently the sufficient uniform operator-norm test $\|V_i\|_{\mathrm{op}}\le\sqrt{3490/2400}\approx1.20589$ from Section~13.4, if imposed on nonsingular integer lifts, already forces signed permutations. It does not leave a continuous family of mildly non-isometric integer matrices. This does not rule out other transports: the test is sufficient, actual key vectors can occupy favourable directions, and re-centring after modular reduction is not necessarily the same as an unreduced linear integer lift.

Exactly norm-preserving integral holonomy lies in a subgroup of the signed-permutation group, intersected with the required $\Rq$-linearity and commutation constraints. The scalar monomials $\pm X^jI_k$ give one cyclic subgroup of order $2n$. They are not a classification of all compatible isometries for every public~$A$.

\begin{proposition}[Exact lattice preservation forces flatness]\label{prop:flat}
Let a smooth real vector bundle have a smoothly locally framed full-rank lattice in every fibre. If a smooth connection's parallel transport along every sufficiently short path preserves these lattices exactly, its curvature is zero.
\end{proposition}
\begin{proof}
Choose a local lattice frame. Transport along a path varying continuously from the constant path gives a continuous family of matrices in the discrete group $GL_N(\ZZ)$, initially the identity. It remains the identity while the path stays in a sufficiently small trivialising neighbourhood. In particular small-loop holonomy is the identity, so the curvature vanishes. Global monodromy around noncontractible loops need not vanish.
\end{proof}

A genuinely curved continuous connection therefore cannot simultaneously preserve a smoothly varying full lattice exactly along every local path under these hypotheses. Preserving only a selected finite graph of paths avoids this conclusion, as can controlled approximation or a different arithmetic model; each changes the problem to be analysed. These obstructions identify design choices rather than proving that curved cryptography is impossible.

\subsection{Sections, noise, exposure and a testable development plan}
On a compact Ricci-flat Kähler manifold with full holonomy $SU(d)$, $d\ge2$, the Bochner principle excludes nonzero global holomorphic tangent vector fields. It does not exclude smooth tangent sections or prescribed values at finitely many points: disjoint local charts and bump functions can interpolate any such finite assignment. The failure of the earlier cryptosystem was its exposed scalar encoding, not a general prohibition on using tangent fibres for keys. An auxiliary bundle is a design choice with its own properties to justify.

For a proposed transport construction retain the actual realised integer norms $S_i$, not merely their worst operator-norm estimate. Conditional on keys, plaintext and trajectory, with fresh independent block coins as in Lemma~\ref{lem:fixedkey}, one sufficient bound is
\begin{equation}\label{eq:realised}
\Pr(\text{stream failure}\mid\text{fixed trajectory and keys})\le\min\Bigl(1,\ 2n\sum_{i=0}^{L-1}\exp\Bigl(-\frac{832^2}{\eta(S_i+1)}\Bigr)\Bigr).
\end{equation}
It reduces to \eqref{eq:budget} when $S_i\le S_{\max}$. Failure to meet \eqref{eq:budgetcond} means only that this particular sufficient certificate is unavailable; it does not establish a decoding failure. If the geometry introduces new error terms or correlations with fresh coins, \eqref{eq:realised} must itself be re-derived. Bounding an exceptional-family event and conditioning on it remain separate from altering setup by rejection sampling.
\begin{enumerate}
\item \emph{Fix a geometric object.} Start with the degree-one bundle over $\mathbb{CP}^1$ for topological nontriviality, or a specified unitary representation on a higher-genus curve for monodromy. State which invariant and connection the experiment uses.
\item \emph{Build the arithmetic interface.} Give charts, cocycles, finite sampling points, integer representatives and a reduction or quantisation rule. Check \eqref{eq:chart}--\eqref{eq:dual} and whether the resulting object is merely publicly conjugate to the flat baseline.
\item \emph{Control distributions and correctness.} Specify the joint key/noise law. Prove a tail estimate with all transported and quantisation errors; compare the realised-norm certificate with the uniform norm test.
\item \emph{Analyse exposure and composition anew.} If $s^{(i)}=V_is^\star$ with public invertible $V_i$, one exposed key reveals $s^\star$. Such a family cannot inherit independent-key exposure protection. Any proposed repair needs its own distribution and reduction.
\item \emph{Demonstrate a benefit.} Compare security or a concrete functional capability at matched key size, ciphertext size, failure budget, memory and legitimate computation. Curvature, nontrivial topology and nonabelian holonomy alone establish no additional attack cost.
\end{enumerate}
The target is an explicit arithmetic construction in which retained geometric data changes a useful, measurable property. Establishing that target would be a further result; the present flat scheme supplies a baseline and falsifiable tests for the programme.

\clearpage
\phantomsection

\end{document}